%% file: ALVar-2022.tex
\documentclass{imsart}

\usepackage[colorlinks,citecolor=blue,urlcolor=blue]{hyperref}

\usepackage[colorlinks,citecolor=blue,urlcolor=blue]{hyperref}

\usepackage{amsmath}
\usepackage{graphicx}
\usepackage{enumerate}
\usepackage[numbers]{natbib}
\usepackage{url}
\usepackage{algorithm}
\usepackage{algpseudocode}
\usepackage{amssymb}
\usepackage{xargs}
\usepackage{ifthen}
\usepackage{bbm}
\usepackage{aliascnt}
\usepackage{stmaryrd}
\usepackage{mathtools}
\mathtoolsset{showonlyrefs}
\usepackage{float}
\usepackage{ushort}
\usepackage{amsthm}
\usepackage{secdot}
\usepackage{subcaption}
\usepackage{xcolor}
\usepackage{comment}
\usepackage{booktabs}
\usepackage{mathtools}
\usepackage{todonotes}
\usepackage{bibentry}
\usepackage{upgreek}
\usepackage{soul}
\newcommand{\ignore}[1]{}
\newcommand{\nobibentry}[1]{{\let\nocite\ignore\bibentry{#1}}}

\include{conv_def}

\usepackage{xspace}

\theoremstyle{definition}

\numberwithin{equation}{section}

\begin{document}

\begin{frontmatter}
	\title{Adaptive online variance estimation in particle filters: the ALVar estimator}
	\runtitle{Adaptive online variance estimation in particle filters}
	
	\begin{aug}
		\author[A]{\fnms{Alessandro} \snm{Mastrototaro}\ead[label=e1]{alemas@kth.se}}
		\and
		\author[A]{\fnms{Jimmy} \snm{Olsson}\ead[label=e2]{jimmyol@kth.se}}
	
		\address[A]{Department of Mathematics, KTH Royal Institute of Technology, Stockholm, \printead{e1,e2}}
	\end{aug}
	
	\begin{abstract}
	We present a new approach---the \namealgo estimator---to estimation of asymptotic variance in sequential Monte Carlo methods, or, particle filters. The method, which adjusts adaptively the lag of the estimator proposed in [Olsson, J. and Douc, R. (2019). Numerically stable online estimation of variance in particle filters. \emph{Bernoulli}, \textbf{25}(2), pp.~1504--1535] applies to very general distribution flows and particle filters, including auxiliary particle filters with adaptive resampling. The algorithm operates entirely online, in the sense that it is able to monitor the variance of the particle filter in real time and with, on the average, constant computational complexity and memory requirements per iteration. Crucially, it does not require the calibration of any algorithmic parameter. Estimating the variance only on the basis of the genealogy of the propagated particle cloud, without additional simulations, the routine requires only minor code additions to the underlying particle algorithm. Finally, we prove that the \namealgo estimator is consistent for the true asymptotic variance as the number of particles tends to infinity and illustrate numerically its superiority to existing approaches. 
	
	\end{abstract}
	
	
\end{frontmatter}

\section{Introduction}\label{sec:intro}
\input{Section_intro.tex}

\section{Preliminaries}\label{sec:prel}
\input{Section_preliminaries.tex}

\section{Main results}\label{sec:results}
\input{Section_results.tex}

\section{Numerical illustrations}\label{sec:simul}
\input{Section_simulations.tex}
\begin{funding}
    This work is supported by the Swedish Research Council, Grant 2018-05230.
\end{funding}


\input{ALVar-2022.bbl}
\appendix
\input{Appendix}

\end{document}

%% file: conv_def.tex
\newtheorem{theorem}{Theorem}[section]
\newaliascnt{proposition}{theorem}
\newtheorem{proposition}[proposition]{Proposition}
\aliascntresetthe{proposition}

\newaliascnt{lemma}{theorem}
\newtheorem{lemma}[lemma]{Lemma}
\aliascntresetthe{lemma}

\newaliascnt{corollary}{theorem}
\newtheorem{corollary}[corollary]{Corollary}
\aliascntresetthe{corollary}

\theoremstyle{definition}
\newaliascnt{definition}{theorem}
\newtheorem{definition}[definition]{Definition}
\aliascntresetthe{definition}

\newtheorem{example}{Example}

\newaliascnt{remark}{theorem}

\aliascntresetthe{remark}

\newtheorem{assumption}{Assumption}

\newcounter{hypA}

\newcounter{hypB}

\def\equationautorefname~#1\null{%
  Equation~(#1)\null
}

\newcommand{\1}[1]{\mathbbm{1}_{#1}}
\DeclareMathOperator*{\argmax}{arg\,max}
\DeclareMathOperator*{\argmin}{arg\,min}

\newcommandx\A[2][1=]{
\ifthenelse{\equal{#1}{}}
{\hspace{-1mm}(\textbf{A\ref{#2}})\hspace{-1mm}}
{\hspace{-1mm}(\textbf{A\ref{#1}--\ref{#2}})\hspace{-1mm}}
}
\newcommand{\af}[1]{h_{#1}} 
\newcommand{\alg}[1]{\mathcal{#1}} 
\newcommand{\am}[1]{\vartheta_{#1}} 
\newcommand{\amm}[1]{\boldsymbol{\vartheta}_{#1}}

\newcommand{\adpf}{\mathsf{AdaPF}}

\newcommandx{\arr}[2][1=]{
\ifthenelse{\equal{#1}{}}
{\upsilon_{\N}^{#2}}
{(\upsilon_{\N}^{#2})^{#1}}
}
\newcommandx{\arrterm}[3][1=]{
\ifthenelse{\equal{#1}{}}
{\tilde{\upsilon}_{\N}(#3,#2)}
{\tilde{\upsilon}_{\N}^{#1}(#3,#2)}
}

\newcommandx{\asvarFFBSm}[4][1=]{
\ifthenelse{\equal{#1}{}}
{\tilde{\sigma}_{#2} \langle #3, #4 \rangle(\af{#1})}
{\tilde{\sigma}_{#2}^2 \langle #3, #4 \rangle(\af{#2})}
}
\newcommandx{\asvarstd}[2][1=]{
\ifthenelse{\equal{#1}{}}
{\sigma_{#2}(\af{#2})}
{\sigma_{#2}^2(\af{#2})}
}
\newcommandx{\asvarFFBSmstd}[2][1=]{
\ifthenelse{\equal{#1}{}}
{\tilde{\sigma}_{#2}(\af{#2})}
{\tilde{\sigma}_{#2}^2(\af{#2})}
}
\newcommand{\asvar}[1]{\sigma_{#1}^2}
\newcommand{\asvarm}[1]{\boldsymbol{\sigma}_{#1}^2}
\newcommand{\asvarp}[1]{\hat{\sigma}_{#1}^2}

\newcommandx\B[2][1=]{
\ifthenelse{\equal{#1}{}}
{\hspace{-1mm}(\textbf{B\ref{#2}})\hspace{-1mm}}
{\hspace{-1mm}(\textbf{B\ref{#1}--\ref{#2}})\hspace{-1mm}}
}
\newcommandx{\BF}[3][1=]{
\ifthenelse{\equal{#1}{}}
{\kernel{D}_{#2, #3}}
{\kernel{D}_{#2, #3}^{#1}}
}
\newcommandx{\BFcent}[3][1=]{
\ifthenelse{\equal{#1}{}}
{\tilde{\kernel{D}}_{#2, #3}}
{\tilde{\kernel{D}}_{#2, #3}^{#1}}
}
\newcommandx{\BFm}[2]{
	\boldsymbol{\mathcal{D}}_{#1, #2}
}
\newcommandx{\BFcentm}[2]{
	\tilde{\boldsymbol{\mathcal{D}}}_{#1, #2}
}

\newcommandx{\bk}[2][1=]{ 
\ifthenelse{\equal{#1}{}}
{\overleftarrow{\kernel{Q}}_{#2}}
{\overleftarrow{\kernel{Q}}_{#2}^{#1}}
}

\newcommand{\bmf}[1]{\set{F}(#1)} 
\newcommand{\borel}[1]{\mathcal{B}(\set{#1})} 

\newcommand{\catdist}{\mathsf{Cat}}
\newcommandx{\cexp}[3][1=]{
\ifthenelse{\equal{#1}{}}
{\mathbb{E}\left[ #2 \mid #3 \right]} 
{\mathbb{E}[ #2 \mid #3 ]}
}
\newcommand{\CLE}{CLE\xspace}

\newcommand{\convd}{\overset{\mathcal{D}}{\longrightarrow}}
\newcommand{\convp}{\overset{\prob}{\longrightarrow}}

\newcommand{\E}{\mathbb{E}}
\newcommand{\enoch}[2]{E_{#1}^{#2}} 
\newcommand{\enochbar}[2]{\bar{E}_{#1}^{#2}}
\newcommand{\eg}{\emph{e.g.}\xspace}
\newcommand{\epart}[2]{\xi_{#1}^{#2}}
\newcommand{\epartm}[2]{\boldsymbol{\xi}_{#1}^{#2}}

\newcommand{\epartbar}[2]{\bar{\xi}_{#1}^{#2}}
\newcommand{\eqdef}{\coloneqq} 
\newcommand{\eqdist}{\overset{\mathcal{D}}{=}}

\newcommand{\ess}{\mathsf{ESS}}

\newcommand{\func}[1]{h_{#1}}
\newcommand{\funcm}[1]{\boldsymbol{h}_{#1}}
\newcommand{\funcbar}[1]{\bar{h}_{#1}}

\newcommand{\g}[1]{g_{#1}} 
\newcommand{\gbar}[1]{\bar{g}_{#1}}

\newcommandx{\genfd}[1][1=]{
	\ifthenelse{\equal{#1}{}}
	{\mathcal{F}}
	{\mathcal{F}_{\N}}
}

\newcommand{\hkbar}[1]{\bar{\kernel{M}}_{#1}}

\newcommand{\hk}[1]{\kernel{M}_{#1}} 

\newcommand{\I}[2]{I_{#1}^{#2}} 
\newcommand{\Ibar}[2]{\bar{I}_{#1}^{#2}}
\newcommand{\ie}{\emph{i.e.}\xspace}
\newcommand{\intvect}[2]{\llbracket #1, #2 \rrbracket}

\newcommand{\kernel}[1]{\mathbf{#1}}
\newcommand{\kletter}{M}
\newcommandx{\K}[1][1=]{
	\ifthenelse{\equal{#1}{}}{{\kletter}}{{M^{#1}}}
}

\newcommandx{\lebfun}[1][1=]{
\ifthenelse{\equal{#1}{}}
{\lebfunletter}
{\lebfunletter_{#1}}
}

\newcommand{\lebfunletter}{\varphi}
\newcommandx{\lk}[2][2=]{ 
\ifthenelse{\equal{#2}{}}
{\kernel{L}_{#1}}
{\kernel{L}_{#1,#2}}
 }
\newcommand{\lkm}[1]{\boldsymbol{\mathcal{L}}_{#1}}

\newcommand{\meas}[1]{\mathsf{M}(#1)}

\newcommand{\emse}{\mathsf{MSE}^N}

\newcommand{\mx}{\text{max}}

\newcommand{\N}{N}
\newcommand{\namealgo}{\textsf{ALVar}\xspace}

\newcommand{\nset}{\mathbb{N}}
\newcommand{\nsetpos}{\mathbb{N}^*}

\newcommand{\ordo}{\mathcal{O}}
\newcommandx{\oscn}[2][1=]{
\ifthenelse{\equal{#1}{}}{\operatorname{osc}(#2)}{\operatorname{osc}^{#1}(#2)}
}

\newcommand{\partfilt}[1]{\mathcal{F}_{#1}^{\N}}

\newcommand{\PF}{\mathsf{PF}}
\newcommandx\post[2][1=]{
\ifthenelse{\equal{#1}{}}
	{\phi_{#2}}
	{\phi_{#2}^{#1}}
}
\newcommandx\postm[2][1=]{
	\ifthenelse{\equal{#1}{}}
	{\boldsymbol{\phi}_{#2}}
	{\boldsymbol{\phi}_{#2}^\N}
}
\newcommandx\postp[2][1=]{
	\ifthenelse{\equal{#1}{}}
	{\pi_{#2}}
	{\pi_{#2}^\N}
}
\newcommandx\postpm[2][1=]{
	\ifthenelse{\equal{#1}{}}
	{\boldsymbol{\pi}_{#2}}
	{\pi_{#2}^\N}
}
\newcommandx\postbar[2][1=]{
	\ifthenelse{\equal{#1}{}}
	{\bar{\phi}_{#2}}
	{\bar{\phi}_{#2}^{#1}}
}
\newcommandx\postpbar[2][1=]{
	\ifthenelse{\equal{#1}{}}
	{\bar{\pi}_{#2}}
	{\bar{\pi}_{#2}^{#1}}
}

\newcommand{\predkerbar}[2]{\bar{\kernel{L}}_{#1:#2}}

\newcommand{\prob}{\mathbb{P}} 

\newcommand{\probmeas}[1]{\mathsf{M}_1(#1)}
\newcommand{\proj}{\boldsymbol{\Pi}}
\newcommand{\propker}[1]{\kernel{P}_{#1}}

\newcommand{\propkerm}[1]{\boldsymbol{\mathcal{P}}_{#1}}

\newcommandx\nres[2][1=]{
	\ifthenelse{\equal{#1}{}}
	{r_{#2}}
	{r_{#2}^{#1}}
}
\newcommandx\res[2][1=]{
	\ifthenelse{\equal{#1}{}}
	{\rho_{#2}}
	{\rho_{#2}^{#1}}
}

\newcommand{\rn}[1]{\gamma_{#1}}
\newcommand{\rnm}[1]{\boldsymbol{\gamma}_{#1}}
\newcommand{\rset}{\mathbb{R}}
\newcommand{\rsetnonneg}{\mathbb{R}_+}
\newcommand{\rsetpos}{\mathbb{R}_+^*}

\newcommand{\set}[1]{\mathsf{#1}} 

\newcommand{\tensprod}{\varotimes}
\newcommand{\testfsymb}{f}
\newcommandx{\testf}[1][1=]{  
\ifthenelse{\equal{#1}{}}{\testfsymb}{\testfsymb_{#1}}
}
\newcommand{\testfpsymb}{\tilde{f}}
\newcommandx{\testfp}[1][1=]{  
\ifthenelse{\equal{#1}{}}{\testfpsymb}{\testfpsymb_{#1}}
}
\newcommandx{\testfbar}[1][1=]{  
\ifthenelse{\equal{#1}{}}{\bar{f}}{\bar{f}_{#1}}
}

\newcommand{\tstatletter}{\kernel{T}}
\newcommandx\tstat[2][1=]{
\ifthenelse{\equal{#1}{}}
	{\tstatletter_{#2}}
	{\tau_{#2}^{#1}}
}
\newcommandx\tstatm[2][1=]{
\ifthenelse{\equal{#1}{}}
	{\boldsymbol{\tstatletter}_{#2}}
	{\boldsymbol{\tau}_{#2}^{#1}}
}

\newcommand{\update}[1]{#1}

\newcommand{\var}{\text{Var}}

\newcommand{\wgt}[2]{\omega_{#1}^{#2}}
\newcommand{\wgtm}[2]{\boldsymbol{\omega}_{#1}^{#2}}
\newcommand{\wgtsum}[1]{\Omega_{#1}}
\newcommand{\wgtsumm}[1]{\boldsymbol{\Omega}_{#1}}

\newcommand{\wgtbar}[2]{\bar{\omega}_{#1}^{#2}}
\newcommand{\wgtsumbar}[1]{\bar{\Omega}_{#1}}

\newcommand{\Xbar}[1]{\bar{\set{X}}_{#1}}
\newcommand{\Xalgbar}[1]{\bar{\alg{X}}_{#1}}
\newcommand{\xbar}[1]{\bar{x}_{#1}}
\newcommand{\Xinit}{\chi}
\newcommand{\Xinitprop}{\nu}
\newcommand{\Xinitm}{\boldsymbol{\chi}}
\newcommand{\Xinitpropm}{\boldsymbol{\nu}}
\newcommand{\Xpalg}[1]{\boldsymbol{\mathcal{X}}_{\!#1}}
\newcommand{\Xp}[1]{\boldsymbol{\mathsf{X}}_{#1}}

\newcommand{\xvec}[1]{\boldsymbol{x}_{#1}}


%% file: Section_intro.tex
In this paper we present an adaptive online algorithm estimating the asymptotic variance in \emph{particle filters}, or, \emph{sequential Monte Carlo} (SMC) \emph{methods}. SMC methods approximate a given sequence of distributions by propagating recursively a sample of random simulations, so-called \emph{particles}, with associated importance weights. Applications include finance, signal processing, robotics, biology, and several more; see, \eg, \cite{doucet:defreitas:gordon:2001,chopin:papaspiliopoulos:2020}. This methodology, introduced first by \citet{gordon:salmond:smith:1993} in the form of the \emph{bootstrap particle filter}, revolves around two operations: a \emph{selection} step, which resamples the particles in proportion to their importance weights, and a \emph{mutation} step, which randomly propagates the particles in the state space. 

Since the the bootstrap particle filter was introduced, several theoretical results describing the convergence of SMC methods as the number of particles tends to infinity have been established; see, \eg, \cite{cappe:moulines:ryden:2005,delmoral:2004,delmoral:2013}. A contribution of vital importance was made by \citet{delmoral:guionnet:1999}, who established, under general assumptions, a \emph{central limit theorem} (CLT) for standard SMC methods, a result that was later refined in, among others, \cite{chopin:2004,kunsch:2005,douc:moulines:2008}. CLTs are generally essential in Monte Carlo simulation, as these allow the accuracy of produced estimates to be assessed in terms of confidence bounds. However, in the case of particle filters, the asymptotic variance of the weak, Gaussian limit is generally intractable due to the recursive nature of these algorithms. Thus, to estimate the variance of SMC methods is a very challenging task, and although the literature on SMC is vast, only very few works are dedicated to this specific problem. Until just a couple of years ago, the only possible way to estimate the particle-filter variance was take a naive---and computationally very  demanding---approach consisting of calculating the sample variance across independent replicates of the particle filter; see \cite{crisan:miguez:rios:2018} for a similar procedure in the context of parallelisation of SMC methods. An important step towards online variance estimation in particle filters was taken by \citet{chan:lai:2013}, who developed a consistent asymptotic-variance estimator which can be computed sequentially on the basis of a single realisation of the particle filter and without significant additional computational effort. In the same work, the estimator, which we will refer to as the \emph{Chan and Lai estimator} (\CLE), was also shown to be asymptotically consistent as the number of particles tends to infinity. The \CLE was later refined and analysed further in \cite{lee:whiteley:2018,du:guyader:2021}.

In a particle filter, the repeated resampling operations induce genealogical relations between the particles, allowing the estimator---the weighted empirical measure formed by the particles---to be split into terms corresponding to particle subpopulations obtained by stratifying the particle sample by the time-zero ancestors. At each iteration, the \CLE is, simply put, given by the sample variance of these contributions with respect to the average of the full population. However, as time increases, the set of distinct time-zero ancestors depletes gradually, and eventually all the particles share one and the same time-zero ancestor. This \emph{particle-path degeneracy phenomenon} makes the \CLE collapse to zero in the long run. In order to remedy to this issue and to push the technology towards truly online settings, \citet{olsson:douc:2017} devised a lag-based, numerically stable strategy in which the particle sample at time $n$ is stratified by the ancestors at some more recent time $(n-\lambda) \vee 0$, where $\lambda \in \nset$ is a fixed lag parameter. Such a procedure---which can still be implemented in an online fashion---avoids completely the issue of particle-path degeneracy at the cost of a bias induced by the lag. Still, under mild assumptions being satisfied also for models with a non-compact state space, the authors managed to bound this bias uniformly in time by a quantity that decays geometrically with $\lambda$. The simulation study presented in the same work confirms the long-term stability of the produced estimates, which stay, when the lag is well chosen, very close to the ones produced by the naive estimator for arbitrarily long periods of time. However, designing the lag parameter $\lambda$ is highly non-trivial as the optimal choice  depends on the ergodicity properties of the model; indeed, the user faces a delicate bias--variance tradeoff in the sense that using a too small lag results in a numerically stable but significantly biased estimator, while using a too large one eliminates the bias at the cost of high variance implied by the same degeneracy issue as that of the \CLE. 

In this paper we develop further the lag-based approach of \citet{olsson:douc:2017} and propose an estimator that is capable of adapting automatically, by monitoring the degree of depletion of the ancestor sets, the size of the lag as the particles evolve. Like the fixed-lag method in \cite{olsson:douc:2017}, our \emph{adaptive-lag variance} (\namealgo) \emph{estimator} operates online with time-constant memory requirements, but does not require the calibration of any algorithmic parameter. Moreover, estimating the variance only on the basis of the genealogy of the propagated particle cloud, without additional simulations, the routine requires only minor code additions to the underlying particle algorithm \update{and has a linear computational complexity in the number of particles that is fully comparable to the particle filter itself. 
These appealing complexity properties are absolutely crucial in practical applications. As a comparison, the online approach to variance estimation in SMC methods recently proposed in \cite{janati:lecorff:petetin:2022}, relying on backward-sampling techniques, has, at best, a quadratic complexity in the number of particles, which is impractical for large particle sample sizes.} Unlike previous works on variance estimation in SMC, which focus on the standard \emph{bootstrap particle filter} operating on \emph{Feynman--Kac models} \cite{delmoral:2004}, our estimator applies to more general \emph{auxiliary particle filters} (APF) \cite{pitt:shephard:1999} and classes of models. In this setting, we show that the \namealgo estimator is asymptotically consistent as the number of particles tends to infinity. Moreover, we claim and illustrate numerically that the values of the lag chosen adaptively by the algorithm stay stable over time and increase, on the average, only logarithmically with the number of particles; the latter property is fundamental to avoid an excessive demand of computational resources in applications. Furthermore, we extend our estimator to particle filters with adaptive resampling, in which the selection operation is performed only when triggered by some criterion monitoring the particle weight degeneracy, yielding the first SMC variance estimator in that context. 

The rest of the paper is structured as follows: in Section~\ref{sec:prel} we introduce some notation, our general model framework, SMC methods, and variance estimation in particle filters; in addition, we show that all the results obtained in the framework of Feynman--Kac models and the bootstrap particle filter can be extended to our framework and the APF. In Section~\ref{sec:results} we present the \namealgo estimator, prove of its consistency, and provide an extension to particle filters with adaptive resampling. Section~\ref{sec:simul} provides numerical simulations illustrating the algorithm on some classic state-space models. 
Finally, Sections~\ref{sec:appendix1}--\ref{sec:appendix2} provide some of the proofs of the results stated in Sections~\ref{sec:prel}--\ref{sec:results}.

%% file: Section_preliminaries.tex
\subsection{Notation}
We denote by $\nset$ the set of nonnegative integer numbers and let $\nsetpos \eqdef \nset\setminus\{0\}$. For every $(m,n)\in \nset^2$ such that $m\le n$, we denote $\intvect{m}{n}\eqdef\{k\in\nset:m\le k\le n\}$. Moreover, we let $\rsetnonneg$ and $\rsetpos$ be the sets of nonnegative and positive real numbers, respectively, and denote vectors by $x_{m:n}\eqdef(x_m,x_{m+1},\dots,x_{n-1},x_n)$. For a finite set $(p_i)_{i=1}^N$, $N \in\nsetpos$, of nonnegative numbers, we denote by $\catdist((p_i)_{i=1}^N)$ the categorical distribution with sample space $\intvect{1}{N}$ and probability function $\intvect{1}{N} \ni i \mapsto p_i/\sum_{\ell=1}^N p_\ell$. For some general state space $(\set{E}, \alg{E})$ we let $\meas{\alg{E}}$ and $\bmf{\alg{E}}$ be the sets of probability measures on $\alg{E}$ and bounded measurable functions on $(\set{E},\alg{E})$, respectively. For any $\mu \in \meas{\alg{E}}$ and $h \in \bmf{\alg{E}}$ we denote by $\mu h \eqdef \int h(x) \, \mu(dx)$ the Lebesgue integral of $h$ with respect to $\mu$.  

The following kernel notation will be frequently used. Let $(\set{E}', \alg{E}')$ be another measurable space; then a (possibly unnormalised) transition kernel $\kernel{K}:\set{E} \times \alg{E}'\rightarrow\rsetnonneg$ induces the following operations. For any $h \in \bmf{\alg{E}'}$ and $\mu \in \meas{\alg{E}}$ we may define the measurable function 
$$
\kernel{K} h: \set{E} \ni x \mapsto \int h(x') \, \kernel{K}(x, dx')
$$
as well as the measures
\[
\begin{split}
\mu \kernel{K} : \alg{E}' \ni A &\mapsto \int \kernel{K}(x, A) \, \mu(dx), \\
\mu \tensprod \kernel{K} : \alg{E} \tensprod \alg{E}' \ni A &\mapsto \iint \1{A}(x, y)  \, \kernel{K}(x, dy) \, \mu(dx). 
\end{split}
\]
Now, let $(\set{E}'', \alg{E}'')$ be a third measurable state-space and $\kernel{L}$ a possibly unnormalised transition kernel on $\set{E}' \times \alg{E}''$; then, similarly to the operations between measures and kernels, we may define the products
\begin{align}
\kernel{K} \kernel{L} : \set{E} \times \alg{E}'' \ni (x, A) &\mapsto \iint \1{A}(z) \, \kernel{K}(x, dy) \, \kernel{L}(y, dz), \\ 
\kernel{K} \tensprod \kernel{L} : \set{E} \times (\alg{E}' \tensprod \alg{E}'') \ni (x, A) &\mapsto \iint \1{A}(y,z) \, \kernel{K}(x, dy) \, \kernel{L}(y, dz).
\end{align}

\subsection{Model setup}
\label{subsec:model}

In order to define the distribution-flow model under consideration, let $(\set{X}_n,\alg{X}_n)_{n\in\nset}$ be a sequence of measurable state spaces. We introduce unnormalised transition kernels $(\lk{n})_{n \in \nset}$, $\lk{n} : \set{X}_n \times \alg{X}_{n+1}\rightarrow\rsetnonneg$, where each $\lk{n}$ is such that $\sup_{x_n \in \set{X}_n} \lk{n} \1{\set{X}_{n+1}}(x_n) < \infty$. For compactness, we write $\lk{k}[m] \eqdef \lk{k}\lk{k+1}\cdots\lk{m}$ whenever $k \le m$, otherwise $ \lk{k}[m]=\mathrm{id}$ by convention. In addition, we let $\Xinit$ be some measure on $\alg{X}_0$. Using these quantities we may define a flow $\post{n} \in \probmeas{\alg{X}_n}$, $n \in \nset$, of probability distributions by letting, for every $n \in \nset$, 
\begin{equation}\label{eq:filtmeas}
\post{n} = \frac{\Xinit \lk{0}[n-1]}{\Xinit\lk{0}[n-1]\1{\set{X}_n}}.
\end{equation}

\begin{example}[Feynman--Kac models]\label{ex:fk}
\emph{Feynman--Kac models} are applied in a variety of scientific fields such as statistics, physics, biology, and signal processing; see \cite{delmoral:2004} for a broad coverage of the topic. For every $n\in\nset$, let $\hk{n}:\set{X}_n \times \alg{X}_{n+1}\rightarrow[0,1]$ be a Markov transition kernel, $\g{n} : \set{X}_n \rightarrow \rsetnonneg$ a measurable potential function and $\Xinitprop$ be a probability measure on $(\set{X}_0,\alg{X}_0)$. Then the \emph{Feynman--Kac model} $(\post{n})_{n \in \nset}$ induced by $\nu$ and $((\set{X}_n, \alg{X}_n), \hk{n}, \g{n})_{n \in \nset}$ is given \eqref{eq:filtmeas} with 
\begin{equation}
	\lk{n}\func{n+1}(x_n)= \hk{n}(\g{n+1} \func{n+1})(x_n),\quad \func{n+1}\in\bmf{\alg{X}_{n+1}}, \quad x_n \in \set{X}_n.\footnote{More precisely, following the terminology of \cite{delmoral:2004}, the distribution flow generated by kernels of this type is referred to as an \emph{updated Feynman--Kac marginal model}. We will however refer to it as simply a `Feynman--Kac model' for simplicity.} 
\end{equation}
\end{example}

\begin{example}[Partially dominated state-space models]\label{ex:ssm}
General \emph{state-space models} (SSMs) constitute an important modeling tool in a diversity of scientific and engineering disciplines; see {\eg} \cite{cappe:moulines:ryden:2005} and the references therein. An SSM consists of a bivariate Markov chain $(X_n, Y_n)_{n \in \mathbb{N}}$ evolving on some measurable product space $(\set{X} \times \set{Y}, \alg{X} \tensprod \alg{Y})$ according to a Markov transition kernel $\kernel{Q}:(\set{X} \times \set{Y})\times (\alg{X} \tensprod \alg{Y})\rightarrow[0,1]$ constructed on the basis of two other Markov kernels $\kernel{M}:\set{X}\times\alg{X}\rightarrow[0,1]$ and $\kernel{G}:\set{X}\times\alg{Y}\rightarrow[0,1]$ as 
$$
\kernel{Q} : (\set{X}\times\set{Y}) \times (\alg{X} \tensprod \alg{Y}) \ni ((x,y), A) \mapsto \kernel{M} \varotimes \kernel{G}(x, A).
$$
The chain is initialised according to $\Xinitprop\varotimes\kernel{G}$, where $\Xinitprop$ is some probability measure on $(\set{X}, \alg{X})$. In this setting, only the process $(Y_n)_{n \in \nset}$ is observed, whereas the process $(X_n)_{n \in \nset}$---referred to as the \emph{state process}---is unobserved and hence referred to as \emph{hidden}. In this construction, it can be shown (see \cite[Section~2.2]{cappe:moulines:ryden:2005} for details), first, that the state process is itself a Markov chain with transition kernel $\kernel{M}$ and, second, that the observations $(Y_n)_{n \in\nset}$ are conditionally independent given $(X_n)_{n \in \nset}$, with marginal \emph{emission distributions} $Y_n \sim \kernel{G}(X_n, \cdot)$, $n \in \nset$. We assume that the model is \emph{partially dominated}, \ie, that the kernel $\kernel{G}$ admits a transition density $g:\set{X}\times\set{Y}\rightarrow\rsetnonneg$ with respect to some reference measure $\mu$.

Many practical applications of SSMs call for computation of flows of hidden-state posteriors given a sequence $(y_n)_{n \in \nset}$ of observations. In particular, the flow $(\post{n})_{n \in \nset}$ of \emph{filter distributions}, each filter $\post{n}$ being the conditional distribution of the state $X_n$ at time $n$ given $Y_{0:n} = y_{0:n}$, can be expressed as a Feynman--Kac model with $(\set{X}_n, \alg{X}_n)=(\set{X},\alg{X})$, $\hk{n}=\kernel{M}$, and $\g{n}(x)\eqdef g(x,y_n)$ for all $n \in \nset$;  see \cite[Section~3.1]{cappe:moulines:ryden:2005} for details. Inspired by this terminology, we will sometimes refer to each distribution $\post{n}$ in the general flow defined by \eqref{eq:filtmeas} as the \emph{filter at time} $n$. 
\end{example}

\subsection{Sequential Monte Carlo methods}\label{subsec:smc}
In the following we assume that all random variables are well defined on a common probability space $(\Omega,\alg{F},\prob)$. As mentioned in the introduction, we may approximate recursively the distribution sequence $(\post{n})_{n\in \nset}$ by propagating a random sample $(\epart{n}{i},\wgt{n}{i})_{i=1}^N$ of particles and associated weights. Here $N \in \nsetpos$ is the Monte Carlo sample size. More precisely, at each time step, the filter distribution $\post{n}$ is approximated by the weighted empirical measure 
\update{$$
\post[N]{n} \eqdef \sum_{i=1}^N\frac{\wgt{n}{i}}{\wgtsum{n}}\delta_{\epart{n}{i}},  
$$}
where $\Omega_n\eqdef\sum_{i=1}^N\wgt{n}{i}$ and $\delta_{\epart{n}{i}}$ is the Dirac measure located at $\epart{n}{i}$. The APF propagates the sample $(\epart{n}{i},\wgt{n}{i})_{i=1}^N$ recursively as follows. The algorithm is initialised by standard importance sampling, drawing $(\epart{0}{i})_{i=1}^n$ from $\Xinitprop^{\tensprod N}$, where $\Xinitprop \in \probmeas{\alg{X}_0}$ is some proposal distribution dominating $\Xinit$, and letting $\wgt{0}{i}\gets \rn{-1}(\epart{0}{i})$ for each $i$, where $\rn{-1}$ is the Radon--Nikodym derivative of $\Xinit$ with respect to $\Xinitprop$. The auxiliary functions $(\am{n})_{n\in\nset}$, where $\am{n}\in\bmf{\alg{X}_n}$, are introduced in order to favor the resampling of particles that are more likely to be propagated into regions of high likelihood (as measured by the target distributions). The particles are propagated according to some proposal Markov transition kernels $\propker{n}$, $n\in\nset$. These kernels are such that, for each $n\in\nset$ and $x_n\in\set{X}_n$, the probability measure $\propker{n}(x_n,\cdot)$ is absolutely continuous with respect to the the measure $\lk{n}(x_n,\cdot)$. Hence, given $x_n$, 
there is a Radon--Nikodym derivative $\rn{n}(x_n,\cdot)$ such that for every $x_n\in\set{X}_n$ and  $h\in\bmf{\alg{X}_{n+1}}$,
\begin{equation}
\int h(x_{n+1})\, \lk{n}(x_n,dx_{n+1})=\int h(x_{n+1}) \rn{n}(x_n,x_{n+1})\,\propker{n}(x_n,dx_{n+1}).
\end{equation}
Algorithm~\ref{algo:apf} shows one iteration of the APF. In the following we will express one iteration of the APF as $(\epart{n+1}{i},\wgt{n+1}{i},\I{n+1}{i})_{i=1}^N\leftarrow\PF((\epart{n}{i},\wgt{n}{i})_{i=1}^N)$, where also the resampled indices are included in the output for reasons that will be clear later. 
\begin{algorithm}[htb]
	\caption{Auxiliary particle filter}\label{algo:apf}
	\begin{algorithmic}[1]
		\Statex Input: $(\epart{n}{i},\wgt{n}{i})_{i=1}^\N$.
		\Statex Output $(\epart{n+1}{i},\wgt{n+1}{i},\I{n+1}{i})_{i=1}^\N$.
		\For{$i=1\rightarrow\N$}
		\State draw $\I{n+1}{i}\sim \catdist((\wgt{n}{\ell}\am{n}(\epart{n}{\ell}))_{\ell=1}^\N)$;\State draw $\epart{n+1}{i}\sim \propker{n}(\epart{n}{\I{n+1}{i}},\cdot)$;\State weight $\wgt{n+1}{i}\leftarrow  \rn{n}(\epart{n}{\I{n+1}{i}}, \epart{n+1}{i})/\am{n}(\epart{n}{\I{n+1}{i}})$;
		\EndFor
	\end{algorithmic}
\end{algorithm}
As mentioned in the introduction, the first proof of the CLT for SMC methods obtained in \cite{delmoral:guionnet:1999} has been refined and generalised in a number of papers. The following theorem provides a CLT for APFs in the general model context of Section~\ref{subsec:model}, and follows immediately from the more general result \cite[Theorem~B.6]{mastrototaro:olsson:alenlov:2021}.
\footnote{\update{In the mentioned work, where focus is on particle approximation of path-distribution flows using backward-sampling techniques, the underlying model is, in contrast to here, supposed to be dominated in the sense that the kernels $(\lk{n})_{n\in\nset}$ and $(\propker{n})_{n\in\nset}$ are assumed to have densities with respect to some dominating measures. Still, as long as focus is only on the marginal flow $(\post{n})_{n \in \nset}$, the same proofs may be straightforwardly adapted to the non-dominated setting.}} 

\begin{assumption}\label{assum:bounded}
For every $n \in \nset$, $\am{n} \in \bmf{\alg{X}_n}$ and $\rn{n}/\am{n} \in \bmf{\alg{X}_n}$. In addition, $\rn{-1} \in \bmf{\alg{X}_0}$. 
\end{assumption}
\begin{theorem}\label{prop:clt}
Let Assumption~\ref{assum:bounded} hold. Then for every $n\in\nset$ and $\func{n}\in\bmf{\alg{X}_n}$, as $N\to\infty$,
\begin{equation}\label{eq:CLT}
\sqrt{N}(\post[N]{n}\func{n}-\post{n}\func{n})\convd \sigma_n(\func{n})Z,
\end{equation}
with $Z$ being standard normally distributed and $ \asvar{n}(\func{n}) \eqdef\asvar{0,n}(\func{n})$, where, for $\ell\in\intvect{0}{n}$,
\begin{multline}\label{eq:asvar}
\asvar{\ell,n}(\func{n})
\eqdef \frac{\Xinit(\rn{-1}\{\lk{0}[n-1](\func{n}-\post{n}\func{n})\}^2)}{(\Xinit \lk{0}[n-1]\1{\set{X}_n})^2} \1{\{\ell=0\}}
\\+ \sum_{m=(\ell-1)\vee 0}^{n-1} \post{m}\am{m} \frac{\post{m}\lk{m}(\am{m}^{-1}\rn{m} \{\lk{m+1}[n-1](\func{n}-\post{n}\func{n})\}^2  )}{(\post{m} \lk{m}[n-1] \1{\set{X}_n})^2}.
\end{multline}
\end{theorem}
The truncated asymptotic variance ($\ell > 0$) will be useful later on. 

The present paper focuses on estimating online, as $n$ increases and while the particle sample is propagated, the sequence of the asymptotic variances in \eqref{eq:CLT}. Before presenting our online variance estimator, the next section provides a brief overview of some current approaches.

\subsection{Estimation of the asymptotic variance}
\label{eq:estimation:asymptotic:variance}

As touched upon in the introduction, a naive approach to variance estimation in particle filters is to use a brute-force strategy which runs a sufficiently large number $K \in \nsetpos$ of independent particle filters. Then the asymptotic variance of interest can be estimated by multiplying the sample variance of these filter approximations by $N$. \update{However, having $\ordo(KN)$ complexity, where 
$N$ as well as $K$ should be sufficiently large to provide precise filter and variance estimates, respectively,  
this approach is clearly computationally impractical. Moreover, implementing this procedure in an online fashion requires all the samples of each particle filter to be stored, implying also an $\ordo(KN)$ memory requirement.}

Appealingly, the online approach devised \citet{chan:lai:2013} estimates consistently the sequence of asymptotic variances based only on the cloud of evolved particles and without requiring the execution of multiple SMC algorithms in parallel or any additional simulations. 
This is possible by keeping track, as $n$ increases, of the so-called \emph{Eve indices} $(\enoch{n}{i})_{i=1}^N$ (using the terminology of \citet{lee:whiteley:2018}) identifying  the particles at time zero from which the ones at time $n$ originate, in the sense that $\enoch{n}{i}$ denotes the index of the time-zero ancestor of particle $ \epart{n}{i}$. These indices can be traced iteratively in the particle filter by initially letting, for all $i\in\intvect{1}{N}$, $\enoch{0}{i}\leftarrow i$ and then, as $n$ increases, update the same according to $\enoch{n+1}{i}\leftarrow \enoch{n}{\I{n+1}{i}}$. Such updates are straightforwardly implemented by adding one line of code after the selection operation on Line~2 in Algorithm~\ref{algo:apf}. Then the \CLE estimator 
of $ \asvar{n}(\func{n})$ is, for any $\func{n}\in \bmf{\alg{X}_n}$, given by
\begin{equation}\label{eq:CLE}
	\asvarp{n}(\func{n})\eqdef N \sum_{i=1}^N\left(\sum_{j:\enoch{n}{j}=i}\frac{\wgt{n}{j}}{\wgtsum{n}}\{\func{n}(\epart{n}{j})-\post[N]{n}\func{n}\}\right)^2.
\end{equation}
As a main result, \citet{chan:lai:2013} established the consistency of this estimator, in the sense that for every $n\in\nset$ and $\func{n}\in\bmf{\alg{X}_n}$, $\asvarp{n}(\func{n})\convp \asvar{n}(\func{n})$ as $N$ tends to infinity. Although being groundbreaking in theory, the estimator \eqref{eq:CLE} suffers a severe drawback in practice due to the particle-path degeneracy phenomenon. Indeed, because of the resampling operation, at each iteration of the filter some particles will inevitably be propagated from the same parent-particle. Thus, eventually, when $ n $ is large enough, all particles will share the same time-zero ancestor, \emph{i.e.}, there will exist $i_0\in\intvect{1}{N}$ such that $\enoch{n}{i}=i_0$, for all $i\in \intvect{1}{N}$. Recently, \citet{koskela:jenkins:johansen:spano:2020} showed, under some standard mixing assumptions on the model, that the number of iterations needed to make the genealogical paths of the particles coalesce in this way is $\ordo(N)$. Hence, eventually the estimate \eqref{eq:CLE} collapses to zero, which makes it unusable for large values of $n$. In practice, the estimator exhibits poor accuracy and high variability already when the Eve indices take only few distinct values, as the variance estimates will be based on only a few distinct values in that case.


In order to remedy this issue, \citet{olsson:douc:2017} suggest to, rather than tracing the time-zero ancestors, estimate the variance on the basis of the ancestors in some more recent generation. For this purpose, they introduce the \emph{Enoch indices} defined recursively, for all $i\in\intvect{1}{N}$ and $m\in\intvect{0}{n+1}$, by
\begin{equation}\label{eq:enoch-update}
	\enoch{m,n+1}{i}\eqdef\begin{cases}
		i\quad &\text{for }m=n+1,\\\enoch{m,n}{\I{n+1}{i}}&\text{for }m<n+1.
	\end{cases}
\end{equation}
In words, $\enoch{m,n}{i}$ indicates the index of the ancestor at time $m \leq n$ of particle $i$ at time $n$; moreover, notice that when $m=0$, these indices correspond to the Eve indices. Then, letting, for some lag $\lambda\in\nset$, $n\langle\lambda\rangle\eqdef(n-\lambda)\vee 0$, the \CLE \eqref{eq:CLE} is replaced by the modified estimator
\begin{equation}\label{eq:OD}
	\asvarp{n,\lambda}(\func{n})\eqdef N \sum_{i=1}^N\left(\sum_{j:\enoch{n\langle\lambda\rangle,n}{j}=i}\frac{\wgt{n}{j}}{\wgtsum{n}}\{\func{n}(\epart{n}{j})-\post[N]{n}\func{n}\}\right)^2.
\end{equation}
Since for a given lag $ \lambda $, the number $n\langle\lambda\rangle$ of the generation to which the Enoch indices underpinning the estimator \eqref{eq:OD} refer varies with $n$, the algorithm requires the storage and iterative updating of a \emph{window} $(\enoch{n\langle\lambda\rangle,n}{i}, \dots, \enoch{n,n}{i})_{i = 1}^N$ 
of Enoch indices. One iteration of the procedure is shown in Algorithm~\ref{algo:pfenoch}, which is initialised be generating the initial particle cloud as in Algorithm~\ref{algo:apf} and letting, in addition, $\enoch{0,0}{i}\leftarrow i$ for all $i\in\intvect{1}{N}$. We observe that the memory requirement and computational complexity of each iteration of the algorithm are both $\ordo(\lambda N)$, independently of the time index $n$.

\begin{algorithm}[H]
	\caption{APF with online Enoch-indices updates}\label{algo:pfenoch}
	\begin{algorithmic}[1]
		\Statex \textbf{Data}: $(\epart{n}{i},\wgt{n}{i})_{i=1}^N$, $(\enoch{n\langle\lambda\rangle,n}{i}, \dots, \enoch{n,n}{i})_{i = 1}^N$ 
		\Statex \textbf{Result}: $(\epart{n+1}{i},\wgt{n+1}{i})_{i=1}^N$, 
		$(\enoch{(n+1)\langle\lambda\rangle,n+1}{i}, \dots, \enoch{n+1,n+1}{i})_{i=1}^N$
		\State run $(\epart{n+1}{i},\wgt{n+1}{i},\I{n+1}{i})_{i=1}^N\leftarrow\PF((\epart{n}{i},\wgt{n}{i})_{i=1}^N)$;
		\For{$i=1\rightarrow N$}
		\For{$ m=(n+1)\langle\lambda\rangle\rightarrow n $}
		\State set $\enoch{m,n+1}{i}\leftarrow \enoch{m,n}{\I{n+1}{i}}$;
		\EndFor
		\State set $\enoch{n+1,n+1}{i}\leftarrow i$;
		\EndFor
	\end{algorithmic}
\end{algorithm}

The estimator \eqref{eq:OD} is not consistent for the asymptotic variance $\asvar{n}(\func{n})$ as $N$ tends to infinity; still, \citet[Proposition~8]{olsson:douc:2017} showed that for all $\lambda\in\nset$ and $\func{n}\in\bmf{\alg{X}_n}$, $\asvarp{n,\lambda}(\func{n})$ converges to $\asvar{n\langle\lambda\rangle,n}(\func{n})$ in probability as $N$ tends to infinity, where $\asvar{n\langle\lambda\rangle,n}(\func{n})$ is the truncated asymptotic variance given by \eqref{eq:asvar}, a quantity that is always smaller than the true asymptotic variance. Additional theoretical results \citep[Section~4]{olsson:douc:2017} establish that under mild, verifiable model assumptions, the asymptotic bias induced by the truncation decays geometrically fast with $\lambda$ (uniformly in $n$). 

\update{
The results of \cite{olsson:douc:2017} were derived in the context of Feynman--Kac models and standard bootstrap particle filters, which is a more restrictive setting than the one considered here. Still, interestingly, it is   possible to show that a general APF operating on a general distribution-flow in the form  \eqref{eq:filtmeas} can actually be interpreted as standard bootstrap filter operating on a certain auxiliary, extended Feynman--Kac model. Thus, using this trick, which  is described in detail in  Section~\ref{sec:appendix1}, we are able to extend the consistency results obtained in \cite{olsson:douc:2017} to the general setting of the present paper. This is the contents of Theorem~\ref{thm:consistency-fixed-lag}, whose proof is found in Section~\ref{sec:appendix1}.
}

\begin{assumption}\label{assum:1}
	For all $n\in\nset$ and $(x_n,x_{n+1})\in\set{X}_n\times\set{X}_{n+1}$, $$\frac{\rn{n}(x_n,x_{n+1})\am{n+1}(x_{n+1}) }{\am{n}(x_n)}>0$$ and $$\sup_{(x_n,x_{n+1})\in\set{X}_n\times\set{X}_{n+1}}\frac{\rn{n}(x_n,x_{n+1})\am{n+1}(x_{n+1}) }{\am{n}(x_n)}<\infty.$$
	Moreover, for all $x_0\in\set{X}_0$, $$\rn{-1}(x_0)\am{0}(x_0)>0\quad\text{and}\quad\sup_{x_0\in\set{X}_0}\rn{-1}(x_0)\am{0}(x_0)<\infty.$$ 
\end{assumption}
\begin{theorem}\label{thm:consistency-fixed-lag}
	Let Assumptions~\ref{assum:bounded}~and~\ref{assum:1} hold. Then for every $n\in\nset$, $\lambda\in \nset$, and $\func{n} \in\bmf{\alg{X}_n}$, as $N \to \infty$,
	\begin{equation}
	\asvarp{n,\lambda}(\func{n}) \convp \asvar{n\langle\lambda\rangle,n}(\func{n}).
	\end{equation}
\end{theorem}

The main practical issue with the lag-based approach of \cite{olsson:douc:2017} is that the design of an optimal lag might be a difficult task. 
Using a too large lag implies, as for the \CLE, depletion of the set of ancestors supporting the estimator, leading to high variance; on the other hand, using a too small lag decreases this variance, however at the cost of significant underestimation of the asymptotic variance of interest. The fact that the asymptotic bias decreases geometrically fast suggests that we should obtain a good approximation of the asymptotic variance even for moderate values of $\lambda$, but quantifying this optimal lag size may be a laborious task. In the numerical simulations in \cite{olsson:douc:2017}, the algorithm is run multiple times for several distinct values of $\lambda$, whereupon the variance estimates obtained in this manner are compared to that obtained using the naive estimator in order to determine the empirically best lag. This method is not ideal as it requires extensive prefatory computations and does not take into account the possibility of varying the lag as the particles evolve. Instead, it is desirable to keep the lag as large as possible as long as the estimator is of good quality (in some sense) and decrease it whenever we detect some degeneracy, determined by the depletion of the Enoch indices. This argument will be developed further in the next section,  leading to the design of a fully adaptive approach.

%% file: Section_results.tex
\subsection{The \namealgo estimator}
We first need to identify a criterion to determine an optimal lag at a given iteration $n$. We have previously discussed the bias--variance tradeoff, which usually arises when the objective is to minimise the \emph{mean-squared error} (MSE) of an estimator with respect to the estimand of interest. For every $n \in\nset$, $\func{n}\in\bmf{\alg{X}_n}$, and $\lambda\in\nset$, the MSE of the estimator \eqref{eq:OD} can be written as the sum of its variance and its squared bias according to 
\begin{equation}\label{eq:MSE}
\E\left[\left(\asvarp{n,\lambda}(\func{n})-\asvar{n}(\func{n})\right)^2\right]=\var\left(\asvarp{n,\lambda}(\func{n})\right)+\left(\E\left[\asvarp{n,\lambda}(\func{n})\right]-\asvar{n}(\func{n})\right)^2.
\end{equation}
Our intention is to design a routine for adapting the lag $\lambda$ in such a way that \eqref{eq:MSE} is minimised. Even if we do not have closed-form expressions of the expectation and the variance of the lag-based estimator in \eqref{eq:MSE}, we may make the following considerations. 
\begin{itemize}
	\item Since $\asvarp{n,\lambda}(\func{n})$ tends to $\asvar{n\langle\lambda\rangle,n}(\func{n})$ in probability as the number $N$ of particles tends to infinity, where $\asvar{n\langle\lambda\rangle,n}(\func{n})\le\asvar{n}(\func{n})$ for all $\lambda$, and the difference decreases as $\lambda$ approaches $n$, we may assume that even the non-asymptotic bias is reduced when $\lambda$ increases.
	\item On the other hand, the larger the value of $\lambda$ is, the fewer distinct elements has the set $(\enoch{n\langle\lambda\rangle,n}{i})_{i=1}^N$ of Enoch indices, 	causing an increase of the variance of the estimator \eqref{eq:OD}; see Figure~\ref{fig:boxplot}.
\end{itemize}
The reduction of the number of distinct Enoch indices may be tolerated until an increase of the lag is beneficial for the reduction of the bias, but at some point the behavior becomes pathological. Imagine, for instance, that we use the \CLE in the early iterations of the particle filter for estimating the variance; then, at some time $n$, one realises that there exists some $\lambda \in\intvect{0}{n-1}$ for which $\asvarp{n,\lambda}(\func{n}) > \asvarp{n}(\func{n})$, although their asymptotic values are supposed to be in the opposite order and the lag-based estimator is expected to be less variable. This suggests that the Eve indices might be depleted and not reliable anymore for supporting the variance estimator. It is then reasonable to assume that these will be unreliable also in the subsequent steps, since their degeneracy can only get worse. Extending this idea to the Enoch indices, we may define recursively the concept of \emph{depleted Enoch indices}.

\begin{definition}\label{def:depl_enoch}
	Let $(\func{n})_{n\in\nset}$ be a given a sequence of functions such that $\func{n}\in\bmf{\alg{X}_n}$ for all $n$. The Enoch indices $(\enoch{m,n}{i})_{i=1}^N$ are said to be \emph{depleted} if at least one of the following conditions is satisfied:
	\begin{itemize}
		\item[(i)] the Enoch indices $(\enoch{m,n-1}{i})_{i=1}^N$ are depleted;
		\item[(ii)] the Enoch indices $(\enoch{m-1,n}{i})_{i=1}^N$ are depleted and, letting $\lambda:=n-m$, there exists $\lambda'\in\intvect{0}{\lambda-1}$ such that $ \asvarp{n,\lambda}(\func{n})<\asvarp{n,\lambda'}(\func{n}) $.
	\end{itemize}
	By convention, for every $n\in\nset$, the Enoch indices $(\enoch{n,n}{i})_{i=1}^N$ are never depleted, while $(\enoch{-1,n}{i})_{i=1}^N$ are always depleted (even if these are not explicitly defined).
\end{definition} 

\update{In order to check the depletion status of some indices $(\enoch{m,n}{i})_{i=1}^N$ using  Definition~\ref{def:depl_enoch} we need to know the status of previous generations. Thus, in practice, depletion may be determined iteratively forwards in time, starting from $(\enoch{0,0}{i})_{i=1}^N$, which are not depleted by definition. Then for every $n\in\nsetpos$, knowing whether the indices $(\enoch{m,n-1}{i})_{i=1}^N$ are depleted or not for all $m\in\intvect{0}{n-1}$, it is possible to check the same for $(\enoch{m,n}{i})_{i=1}^N$ starting from $m=0$ and proceeding forwards to $m = n$. This is done by checking first condition (i) in Definition~\ref{def:depl_enoch}; if this is not satisfied, then we check condition (ii). The idea behind condition (i) is that if a set $(\enoch{m,n-1}{i})_{i=1}^N$ of Enoch indices is ill-suited to estimate the variance at some time $n-1$, it will not be suited to estimate the variance at any future time, since the number of distinct elements in the set can only decrease with $n$. Regarding condition (ii), if instead the indices $(\enoch{m,n-1}{i})_{i=1}^N$ are non-depleted, we still need to check if there is a more recent generation $(\enoch{m',n}{i})_{i=1}^N$, $m' \in \intvect{m + 1}{n}$, that produces a better estimate. The additional requirement of  $(\enoch{m-1,n}{i})_{i=1}^N$ being depleted serves to guarantee monotonicity, \ie, if $(\enoch{m,n}{i})_{i=1}^N$ are depleted, then $(\enoch{m',n}{i})_{i=1}^N$ should be as well for all $m'\in\intvect{0}{m}$, whereas if $(\enoch{m,n}{i})_{i=1}^N$ are non-depleted, then $(\enoch{m',n}{i})_{i=1}^N$ should not be either for all $m'\in\intvect{m}{n}$.}
\begin{figure}[htb]
	\centering
	\includegraphics[width=\linewidth]{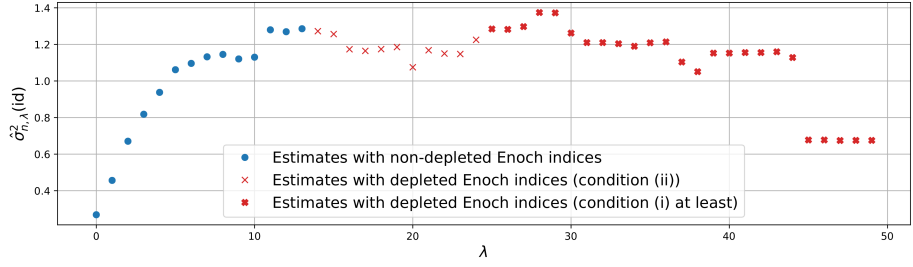}
	\caption{The points in the plot correspond to fixed-lag estimates of $\asvar{n}(\mathrm{id})$\update{, $n$ = 500,} for the stochastic volatility model in Section~\ref{sec:stoc_vol}, with different values of $\lambda$. \update{The particle filter used $N=1000$ particles}. Each estimate is marked differently depending whether the corresponding Enoch indices $(\enoch{n\langle\lambda\rangle,n}{i})_{i=1}^N$ are depleted or not. In addition, we indicate which of the two conditions in Definition~\ref{def:depl_enoch} that indicates non-depletion. For $\lambda \geq 25$, condition (i) was fulfilled at least, while for $\lambda\in\intvect{14}{24}$ condition (ii) (and not (i)) was fulfilled.
	}
	\label{fig:depleted_enoch}
\end{figure}

Algorithm~\ref{algo:adaptive} describes our method, the \emph{adaptive-lag variance} (\namealgo) \emph{estimator}, in which the optimal lag at each iteration is, as established by Theorem~\ref{thm:lambda_largest} below, chosen such that it is the largest one for which the corresponding Enoch indices are not depleted. This non-depletion condition is ensured by selecting recursively $\lambda_{n+1}$ in such a way that it produces the largest estimate, whose selection is bounded from above by $\lambda_n+1$. The lag is initialised by setting $\lambda_0\gets0$.
\begin{algorithm}[htb]
	\caption{\namealgo estimator}\label{algo:adaptive}
	\begin{algorithmic}[1]
		\Statex \textbf{Data}: $(\epart{n}{i},\wgt{n}{i})_{i=1}^N$, $\lambda_n$, $(\enoch{n\langle\lambda_n\rangle,n}{i},\dots, \enoch{n,n}{i})_{i=1}^N$
		\Statex \textbf{Result}: $(\epart{n+1}{i},\wgt{n+1}{i})_{i=1}^N$, $\lambda_{n+1}$, $(\enoch{(n+1)\langle\lambda_{n+1}\rangle,n+1}{i},\dots, \enoch{n+1,n+1}{i})_{i=1}^N$
		\State run $(\epart{n+1}{i},\wgt{n+1}{i},\I{n+1}{i})_{i=1}^N\leftarrow\PF((\epart{n}{i},\wgt{n}{i})_{i=1}^N)$;
		\For{$i=1\rightarrow N$}
		\For{$ m=(n+1)\langle\lambda_n+1\rangle\rightarrow n $}
		\State set $\enoch{m,n+1}{i}\leftarrow \enoch{m,n}{\I{n+1}{i}}$;
		\EndFor
		\State set $\enoch{n+1,n+1}{i}\leftarrow i$;
		\EndFor
		\State set $\displaystyle\lambda_{n+1}\gets\argmax_{\lambda\in\intvect{0}{\lambda_n+1}}\asvarp{n+1,\lambda}(\func{n+1})$; \Comment{each estimator computed according to \eqref{eq:OD}};
	\end{algorithmic}
\end{algorithm}


\begin{theorem}\label{thm:lambda_largest}
    \update{For every $n \in \nset$, let $\lambda_{n}$ be the lag produced by $n$ iterations of Algorithm~\ref{algo:adaptive}. Then if $\lambda_n<n$, none of the Enoch indices $(\enoch{m,n}{i})_{i=1}^N$, $m\in\intvect{n\langle\lambda_n\rangle}{n}$, are depleted whereas all the Enoch indices $(\enoch{m,n}{i})_{i=1}^N$, $m\in\intvect{0}{n\langle\lambda_n\rangle-1}$, are depleted.}  
\end{theorem}

\begin{proof}
    We proceed by induction. The claim is true for $n=0$ since we initialise $\lambda_0\gets 0$. Now, let the claim be true for some $n\in\nset$; then by the induction hypothesis and condition (i) of Definition~\ref{def:depl_enoch}, it holds that $(\enoch{m,n+1}{i})_{i=1}^N$ are depleted for every $m \in \intvect{0}{n \langle \lambda_n \rangle-1}$ if $\lambda_n<n$, where  $n \langle \lambda_n \rangle - 1 = (n + 1) \langle \lambda_n+1\rangle-1$. On the other hand, by the induction hypothesis and the very construction of $\lambda_{n + 1}$ in Algorithm~\ref{algo:adaptive}, none of the depletion conditions of Definition~\ref{def:depl_enoch} are satisfied for $m \in \intvect{(n+1) \langle \lambda_{n+1} \rangle}{n+1}$; hence, the corresponding Enoch indices are not depleted. If $\lambda_{n+1} < \lambda_n+1$, then, again by the construction of $\lambda_{n + 1}$, $(\enoch{m,n+1}{i})_{i=1}^N$ are depleted for $m\in\intvect{(n+1)\langle\lambda_n+1\rangle}{(n+1)\langle\lambda_{n+1}\rangle-1}$ as well by condition (ii). This concludes the proof.
\end{proof}

The computation of the estimator \eqref{eq:OD} has complexity $\ordo(N)$ and is performed $\lambda_n+2$ times at each iteration $n$. In order to have an online algorithm with constant memory requirements we need $\lambda_n$ to be uniformly bounded in $n$. Although in theory the lag might increase indefinitely such that $\lambda_n=n$ for all $n\in\nset$, we may assume that there exists an upper bound on the lag for any fixed number  $N$ of particles. In support of this assumption, we know that the expected number of generations to the time where all the Enoch indices are equal, which is certainly larger than any lag selected by the proposed method, is $\ordo(N)$ uniformly in $n$; see \cite{koskela:jenkins:johansen:spano:2020}. Thus, in practice there will generally exist some $\lambda_\mx$, depending on the model and on $N$ but independent of $n$, such that $\lambda_n<\lambda_\mx$ for all $n\in\nset$. Hence, the final algorithm is \emph{online}, since it has both complexity and memory demand (again dominated by the storage of the Enoch indices) of order $\ordo(\lambda_\mx N)$, independently of $n$, and \emph{adaptive}, since the choice of each new lag is adapted to the output of the particle filter as well as the lag of the previous iteration. In the next section we are going to prove consistency of the estimator and present a heuristic argument concerning the dependence of the lag on the number of particles.

\subsection{Theoretical results}
Next, we show that for every $n \in N$, the resulting adaptive-lag estimator constructed in the previous section is asymptotically consistent for the true asymptotic variance $\asvar{n}(\func{n})$, recalling however that the algorithm is meant to work in the regime where $N$ is fixed and $n$ is arbitrarily large. The `asymptotic' algorithm is not online, since we are going to show that for all $n\in\nset$, $\lambda_n$ tends to $n$ in probability as $N$ grows, implying that in the limit we obtain the CLE at each step. Nevertheless, as we will see later, for a fixed number of particles, the range of the lags returned by the algorithm is expected to grow very slowly with $N$; more precisely, 
in Section~\ref{subsubsec:loglag} we argue for that this range increases only logarithmically with $N$, a claim that is also confirmed by our numerical experiments in Section~\ref{sec:num:lag:analysis}.  

\subsubsection{Consistency}

We now establish the consistency of the \namealgo estimator.
\begin{theorem}\label{thm:conv}
	Let Assumption~\ref{assum:1} hold. For every $n\in\nset$ and $\func{m}\in\bmf{\alg{X}_m}$, $m\in\intvect{1}{n}$, let $(\lambda_m)_{m=1}^n$ be the lags produced by $n$ iterations of Algorithm~\ref{algo:adaptive}. Then, as $N\to\infty$, it holds that $\lambda_n\convp n$ and 
	$
	\asvarp{n,\lambda_n}(\func{n})\convp \asvar{n}(\func{n}).
	$
\end{theorem}
\begin{proof}
We proceed by induction, assuming that the claim holds true for $n-1$. For every $\varepsilon>0$, it holds that 
\begin{multline}
\prob(|\asvarp{n,\lambda_n}(\func{n})-\asvar{n}(\func{n})|\geq 2\varepsilon)\\
\leq \prob(|\asvarp{n,\lambda_n}(\func{n})-\asvarp{n}(\func{n})|\geq \varepsilon)+\prob(|\asvarp{n}(\func{n})-\asvar{n}(\func{n})| \geq \varepsilon), \label{eq:decomp-1}
\end{multline}
where $\asvarp{n}(\func{n})$ is the CLE defined in \eqref{eq:CLE} and based on the same particle system.  The second term on the right-hand side converges to zero as $N \to \infty$, since $ \asvarp{n}(\func{n})=\asvarp{n,n}(\func{n}) $ is consistent for $ \asvar{n}(\func{n})$ by Theorem~\ref{thm:consistency-fixed-lag}. To treat the first term, write 
\begin{equation} \label{eq:decomp-2}
\prob(|\asvarp{n,\lambda_n}(\func{n})-\asvarp{n}(\func{n})| \geq \varepsilon)\le\prob(\asvarp{n,\lambda_n}(\func{n})\ne\asvarp{n}(\func{n})).
\end{equation}
Since $\asvarp{n}(\func{n})=\asvarp{n,n}(\func{n})$, it holds necessarily that $\lambda_n\ne n$ on the event $\asvarp{n,\lambda_n}(\func{n})\ne\asvarp{n}(\func{n})$; thus,
\begin{equation} \label{eq:decomp-3}	
\prob(\asvarp{n,\lambda_n}(\func{n})\ne\asvarp{n}(\func{n}))\le 1-\prob(\lambda_n=n).
\end{equation}
To treat the probability  $\prob(\lambda_n=n)$ we may write 
\begin{align}
	\prob(\lambda_n=n) &= \prob(\lambda_n=n, \lambda_{n-1}={n-1} )+\prob(\lambda_n=n, \lambda_{n-1}<{n-1}) \label{eq:lambda_n:decomp} \\
	&=\prob(\lambda_n=n, \lambda_{n-1}={n-1}), \label{eq:lambda_n:decomp-1}
\end{align}
where the second term of \eqref{eq:lambda_n:decomp} is zero since $\lambda_n\le \lambda_{n-1}+1$ by construction. Now, 
\begin{equation} \label{eq:lambda_n:decomp:alt}
\prob(\lambda_n=n, \lambda_{n-1}={n-1}) = \prob \left( \{\lambda_{n-1}={n-1}\} \bigcap_{\lambda = 0}^{n - 1} \{\asvarp{n,n}(\func{n})\geq \asvarp{n,\lambda}(\func{n})\}  \right), 
\end{equation}
where, by the induction hypothesis,  $\prob(\lambda_{n-1}={n-1})\to 1$ as $N\to \infty$. Moreover, by Theorem~\ref{thm:consistency-fixed-lag}, it holds that $\asvarp{n,\lambda}(\func{n})\convp\asvar{n\langle\lambda\rangle,n}(\func{n})$ and $\asvarp{n,n}(\func{n})\convp\asvar{n}(\func{n})$ as $N\to\infty$, where $\asvar{n\langle\lambda\rangle,n}(\func{n})\le\asvar{n}(\func{n})$ for all $\lambda \in \intvect{0}{n - 1}$, implying that  \eqref{eq:lambda_n:decomp:alt} converges to one as $N\to\infty$. Hence, $\lambda_n\convp n$ and combining this with \eqref{eq:lambda_n:decomp-1}  \eqref{eq:decomp-3}, \eqref{eq:decomp-2}, and  \eqref{eq:decomp-1} 
that $ \asvarp{n,\lambda_n}(\func{n})\convp \asvar{n}(\func{n})$. Finally, the base case holds trivially true since $\lambda_0=0$ and $
\asvarp{0,0}(\func{0})\convp \asvar{0}(\func{0})
$ for all $\func{0}\in\bmf{\alg{X}_0}$.
\end{proof}

\subsubsection{Heuristics on the dependence of the lag on the number of particles}\label{subsubsec:loglag}

In the light of Theorem~\ref{thm:consistency-fixed-lag}, we expect $\lambda_n$ to increase with $N$. It is however crucial to understand how the values of the lags $(\lambda_n)_{n\in\nset}$ depend on $N$, since this will determine the performance and memory requirement of our algorithm. For instance, a linear dependence would imply a quadratic complexity, which is not desirable. In the rest of this section we provide a heuristic showing that if we minimise the MSE \eqref{eq:MSE}, then we may expect $\lambda_n$ to be $\ordo(\log N)$ for all $n\in\nset$. 

If we approximate $\E[\asvarp{n,\lambda}(\func{n})]$ in \eqref{eq:MSE} by the asymptotic limit $\asvar{n\langle\lambda\rangle,n}(\func{n})$, then the second term on the right-hand side is approximately the square of the asymptotic bias,  $
(\asvar{n}(\func{n})-\asvar{n\langle\lambda\rangle,n}(\func{n}))^2$. In \cite{olsson:douc:2017} it is shown that under mild assumptions, the asymptotic bias is $\ordo(\rho^\lambda)$ for some mixing rate $\rho\in(0,1)$. Regarding the variance of $\asvarp{n,\lambda}(\func{n})$, we know that it increases with the lag and therefore as the number of distinct Enoch indices decreases.
Since the variance of a Monte Carlo estimator is generally inversely proportional to the Monte Carlo sample size, we may expect  $\var(\asvarp{n,\lambda}(\func{n}))$ to be $\ordo(1/N_\lambda)$, where $N_\lambda$
is the number of distinct Enoch indices $(\enoch{n\langle\lambda\rangle,n}{i})_{i=1}^N$ at generation $n\langle\lambda\rangle$. Now, by adopting the proof of Corollary~2 in \cite{koskela:jenkins:johansen:spano:2020}, we may argue that under standard mixing assumptions which can be are relaxed in practice, $ N_\lambda$ is $\ordo(N/\lambda)$.\footnote{In order to clarify this reasoning, we redefine $\lambda$ as the random variable such that $n\langle\lambda\rangle$ is the largest time before $n$ for which there are $N' < N$ distinct Enoch indices. Now, in \emph{Kingman’s $ n $-coalescent model}, the expected (continuous) time needed to reach $N'$ distinct ancestors of $N$ offspring is $2/N'-2/N$. If we plug this value (instead of the time to the most recent common ancestor) into the proof of the mentioned corollary, we obtain that $\E[\lambda]$ is $\ordo(N/N')$. Thus, letting $N'=N_\lambda$ we may conclude that $N_\lambda$ is $\ordo(N/\E[\lambda])$.} 
Finally, we determine the order of the optimal lag $\lambda^*$ by letting it be the minimum of the resulting crude approximation 
$$
\lambda \mapsto c \frac{\lambda}{N}+ c' \rho^{2\lambda}
$$
of the MSE \eqref{eq:MSE} as a function of $\lambda$,  
where $c > 0$ and $c' > 0$ are constants independent of $\lambda$ and $N$. It is then easily seen that $\lambda^\ast$ is  
$$
\ordo\left(\frac{1}{2}\log_\rho\left(-\frac{c}{2 c'  N\log\rho}\right)\right)=\ordo(\log_{1/\rho}N)=\ordo(\log N).
$$ 
Although this argument is heuristic, we will see later that it is well supported by our numerical simulations, in which the lags produced are very close the ones minimising the MSE, with a logarithmic dependence on $N$.

\subsection{Extension to particle filters with adaptive resampling}
\label{sec:extension:adaptive:resampling}
We now consider the case in which selection is not necessarily performed at each iteration. Selection is essential in particle filters, as it copes with the well-known importance-weight degeneracy phenomenon (see, \eg,  \cite[Section~7.3]{cappe:moulines:ryden:2005}); however, since resampling adds variance to the estimator, this operation should not be used unnecessarily. A common procedure is hence to resample only when flagged by some weight-degeneracy criterion. One popular such criterion among others is the \emph{effective sample size} (ESS)  \cite{liu:1996} defined by $\ess_n^N \eqdef 1 / \sum_{i=1}^{\N}(\wgt{n}{i}/\wgtsum{n})^2$, which gives an approximation of the number of active particles, \ie, particles with non-degenerated importance weight at time $n$. The ESS is minimal and equal to one when all the weights are equal to zero except one and maximal and equal to $N$ when all weights are non-zero and equal. Using the ESS, one may, \eg, let the resampling operation be triggered only when $\ess_n^N\le \alpha N$, where $\alpha\in(0,1)$ is a design parameter. More generally, we may let $(\res[\N]{n})_{n\in\nset}$ be a sequence of binary-valued random variables indicating whether resampling should be triggered or not. The sequence $(\res[\N]{n})_{n \in \nset}$ is assumed to be adapted to the filtration $(\partfilt{n})_{n \in \nset}$ generated by the particle filter, where $\partfilt{n} \eqdef \sigma((\epart{0}{i})_{i=1}^\N, (\epart{m}{i}, \I{m}{i})_{i=1}^\N, m \in \intvect{1}{n})$. Thus, these indicators may be based on the ESS, letting $\res[\N]{n}=\1{\{ \ess_n^N<\alpha\N\}}$, but also on $n$ only, implying a deterministic selection schedule. Algorithm \ref{algo:adaptivePF} shows one iteration of this adaptive procedure, which we later express in the compact form $(\epart{n+1}{i},\wgt{n+1}{i},\I{n+1}{i})_{i=1}^N\leftarrow\adpf((\epart{n}{i},\wgt{n}{i})_{i=1}^N,\res[N]{n})$. 
\begin{algorithm}[htb]
	\caption{APF with adaptive resampling}\label{algo:adaptivePF}
	\begin{algorithmic}[1]
		\Statex \textbf{Data}: $\{(\epart{n}{i},\wgt{n}{i})\}_{i=1}^N$, $\res[N]{n}$
		\Statex \textbf{Result}: $\{(\epart{n+1}{i},\wgt{n+1}{i},\I{n+1}{i})\}_{i=1}^N$
		\For{$i=1\rightarrow N$}
		\If{$\res[N]{n}=1$}
		\State draw $\I{n+1}{i}\sim \catdist((\wgt{n}{\ell}\am{n}(\epart{n}{\ell}))_{\ell=1}^\N)$;
		\Else
		\State set $\I{n+1}{i}\gets i$;
		\EndIf
		\State draw $\epart{n+1}{i}\sim \propker{n}(\epart{n}{\I{n+1}{i}},\cdot)$;
		\State set $\wgt{n+1}{i}\leftarrow (\wgt{n}{i})^{1-\res[N]{n}} \rn{n}(\epart{n}{\I{n+1}{i}}, \epart{n+1}{i})/(\am{n}(\epart{n}{\I{n+1}{i}}))^{\res[N]{n}}$;
		\EndFor
	\end{algorithmic}
\end{algorithm}
As described in the following, particle filters with adaptive resampling still satisfy a CLT, with asymptotic variance having a structure similar to that of \eqref{eq:asvar} but depending also on an ``asymptotic" resampling schedule to be defined next. 
\begin{assumption}\label{assum:adaptiveRes}
For given $\alpha\in(0,1)$, let $ (\res[\N]{n})_{n\in\nset} $ be defined as 
\begin{equation}\label{eq:rhoN}
	\res[\N]{n} \eqdef \1{\{\ess_n^N<\alpha\N\}}.
\end{equation}
\end{assumption}
The following lemma is adopted from \citep[Lemma 3.5]{mastrototaro:olsson:alenlov:2021} (with $d=\infty$).
\begin{lemma}\label{lemma:ess}
	Let Assumption \ref{assum:adaptiveRes} hold in Algorithm \ref{algo:adaptivePF}. Then for every $n\in\nset$  there exists $\res[\alpha]{n}\in\{0,1\}$ such that, as $\N\to\infty$,
	\begin{equation}
		\res[\N]{n}\convp \res[\alpha]{n}.
	\end{equation}
\end{lemma}

We now have the following CLT for adaptive APFs, whose proof is found in Section~\ref{sec:proof:lemma:CLTadaptive}. 

\begin{theorem}\label{lemma:CLTadaptive}
Let Assumption~\ref{assum:bounded} hold. Let $(\epart{n}{i},\wgt{n}{i})_{i=1}^\N$ be generated by $n$ iterations of Algorithm \ref{algo:adaptivePF} according to a selection schedule $ (\res[\N]{n})_{n\in\nset} $ satisfying Assumption \ref{assum:adaptiveRes}. Then for every $n\in\nset$ and $\func{n}\in\bmf{\alg{X}_n}$, as $N\to \infty$,
\begin{equation}\label{eq:CLTadaptive}
\sqrt{N}(\post[N]{n}\func{n}-\post{n}\func{n})\convd \sigma_n\langle\res[\alpha]{0:n-1}\rangle(\func{n})Z,
\end{equation}
where $Z$ is standard normally distributed random variable and the asymptotic variance $\asvar{n}\langle\res[\alpha]{0:n-1}\rangle(\func{n})$, depending on $\alpha$, is given in detail in Appendix~\ref{sec:appendix2}.
\end{theorem}
When designing a lag-based estimator of the asymptotic variance provided by Theorem~\ref{lemma:CLTadaptive}, it turns out to be more convenient to define the lag in terms of the number of resampling operations rather than the number of iterations of the particle filter.  
For this purpose, let $\nres{n} \eqdef \sum_{m=0}^{n-1}\res[N]{m}$ be the counter of the number of times selection is performed before time $n$ (with the convention $\nres{0}=0$). Then the Enoch indices at each time $n$ will be indexed by the resampling times $\nres{n}$ rather than $n$, since every iteration without resampling leaves these unaltered. More specifically, in the following, a generic Enoch index $\enoch{m,\nres{n}}{i}$ will indicate the ancestor of the particle $\epart{n}{i}$ at any time $n'\in\intvect{0}{n}$ such that $\nres{n'}=m \in \intvect{0}{\nres{n}}$. Then for all $i\in\intvect{1}{N}$ and $m\in\intvect{0}{\nres{n+1}}$, the update \eqref{eq:enoch-update} can been rewritten as
\begin{equation}
\enoch{m,\nres{n+1}}{i} \eqdef
\begin{cases}
i\quad &\text{for }m=\nres{n+1},\\
\enoch{m,\nres{n}}{\I{n+1}{i}}&\text{for }m<\nres{n+1}.
\end{cases}
\end{equation}
Notice that when we do not have resampling at time $n$, it holds that $\nres{n+1}=\nres{n}$ and $\I{n+1}{i}=i$, implying $ \enoch{m,\nres{n}}{i} = \enoch{m,\nres{n+1}}{i} $ for all $m\in\intvect{0}{\nres{n+1}}$ and $i\in\intvect{1}{N}$. In practice, for a given $n\in\nset$, the lag takes on values in $\intvect{0}{\nres{n}}$ instead of $\intvect{0}{n}$ and, as before, the expression $\nres{n}\langle\lambda\rangle$, $\lambda\in\nset$, indicates the quantity $(\nres{n}-\lambda)\vee 0$. In this setting, the estimator \eqref{eq:OD} is rewritten as 
\begin{equation}\label{eq:ODadapt}
	\asvarp{n,\lambda}(\func{n}) \eqdef N\sum_{i=1}^N\left(\sum_{j:\enoch{\nres{n}\langle\lambda\rangle,\nres{n}}{j}=i}\frac{\wgt{n}{j}}{\wgtsum{n}}\{\func{n}(\epart{n}{j})-\post[N]{n}\func{n}\}\right)^2.
\end{equation}
Algorithm~\ref{algo:adaptive2} shows one update of the adaptive-resampling APF along with the calculation of the corresponding \namealgo estimate. Corollary~\ref{cor:convAdaptive}, whose proof is found in Appendix~\ref{subsec:proofCor}, provides the consistency of the variance estimator produced by the algorithm.

\begin{algorithm}[htb]
	\caption{\namealgo estimator for APF with adaptive resampling}\label{algo:adaptive2}
	\begin{algorithmic}[1]
		\Statex \textbf{Data}: $(\epart{n}{i},\wgt{n}{i})_{i=1}^N$, $\res[N]{n}$, $\lambda_n$, $(\enoch{\nres{n}\langle\lambda_n\rangle,{\nres{n}}}{i},\dots, \enoch{{\nres{n}},{\nres{n}}}{i})_{i=1}^N$
		\Statex \textbf{Result}: $(\epart{n+1}{i},\wgt{n+1}{i})_{i=1}^N$, $\lambda_{n+1}$, $(\enoch{\nres{n+1}\langle\lambda_{n+1}\rangle,\nres{n+1}}{i},\dots, \enoch{\nres{n+1},\nres{n+1}}{i})_{i=1}^N$
		\State run $(\epart{n+1}{i},\wgt{n+1}{i},\I{n+1}{i})_{i=1}^N\leftarrow\adpf((\epart{n}{i},\wgt{n}{i})_{i=1}^N,\res[N]{n})$;
		\If{$\res[N]{n}=1$}
		\State set $\nres{n+1}\leftarrow \nres{n}+1$;
		\For{$i=1\rightarrow N$}
		\For{$ m=\nres{n+1}\langle\lambda_n+1\rangle\rightarrow \nres{n} $}
		\State set $\enoch{m,\nres{n+1}}{i}\leftarrow \enoch{m,\nres{n}}{\I{n+1}{i}}$;
		\EndFor
		\State set $\enoch{\nres{n+1},\nres{n+1}}{i}\leftarrow i$;
		\EndFor
		\State set $\displaystyle\lambda_{n+1}\gets\argmax_{\lambda\in\intvect{0}{\lambda_n+1}}\asvarp{n+1,\lambda}(\func{n+1})$; 
		\Comment{each estimator computed according to \eqref{eq:ODadapt}}
		\Else
		\State set $\lambda_{n+1}\gets \lambda_n$;
		\Comment{implication: $(\enoch{\nres{n+1}\langle\lambda_{n+1}\rangle,\nres{n+1}}{i},\dots, \enoch{\nres{n+1},\nres{n+1}}{i})_{i=1}^N=(\enoch{\nres{n}\langle\lambda_n\rangle,{\nres{n}}}{i},\dots, \enoch{{\nres{n}},{\nres{n}}}{i})_{i=1}^N$}
		\EndIf
	\end{algorithmic}
\end{algorithm}

\begin{corollary}\label{cor:convAdaptive}
	Let Assumptions~\ref{assum:1}~and~\ref{assum:adaptiveRes} hold. Moreover, for every $n\in\nset$ and $\func{m}\in\bmf{\alg{X}_m}$, $m\in\intvect{1}{n}$, let $(\lambda_m)_{m=1}^n$ be the lags produced by $n$ iterations of Algorithm~\ref{algo:adaptive2}. Then, letting $ \asvarp{n,\lambda_n}(\func{n}) $ be computed according to \eqref{eq:ODadapt}, it holds, as $N\to\infty$,
	$$
	\asvarp{n,\lambda_n}(\func{n})\convp \asvar{n}\langle\res[\alpha]{0:n-1}\rangle(\func{n}).
	$$
\end{corollary}

%% file: Section_simulations.tex
In this section we apply, as an illustration, our approach to optimal filtering in SSMs. In order to benchmark carefully our variance estimator against the fixed-lag estimator of \cite{olsson:douc:2017}, we tested \namealgo on the same SSMs as in the latter work, namely
\begin{itemize}
	\item[--] the \emph{stochastic volatility model} introduced by \citet{hull:white:1987} and
	\item[--] a \emph{linear Gaussian state space model} for which exact computation of the filter is possible using the Kalman filter.
\end{itemize}

\subsection{Stochastic volatility model}\label{sec:stoc_vol}
Out first SSM is governed by the equations
\begin{equation}\label{eq:stocvol}
\begin{split}
	X_{n+1}&=a X_n +\sigma U_{n+1}, \\
	Y_n&=b \exp(X_n/2) V_n, 
\end{split}
\quad n \in \nset,
\end{equation}
where $(U_n)_{n\in \nsetpos}$ and $(V_n)_{n\in\nset}$ are sequences of uncorrelated standard Gaussian noise variables. The parameters are assumed to be known, with $(a,b,\sigma)=($0.975, 0.641, 0.165). We only observe the process $(Y_n)_{n\in\nset}$, representing stock log-returns, while $(X_n)_{n\in\nset}$, representing the log-volatility, is a hidden state process which we want to infer. The state $X_0$ is initialised according to a zero-mean Gaussian distribution with variance $ \sigma^2/(1-a^2)$, \ie, the stationary distribution of the state process. Thus, we deal with a fully dominated nonlinear SSM with $\set{X}=\set{Y}=\rset$, $\alg{X}=\alg{Y}=\borel{\rset}$, the Borel $\sigma$-field on $\rset$, in which both $\kernel{M}$ and $\kernel{G}$ are Gaussian kernels.

A record $y_{0:5000}$ of observations was obtained by simulating the process  $(X_n, Y_n)_{n\in \mathbb{N}}$ under the dynamics \eqref{eq:stocvol} and the given parameterisation. 
For all $n\in\nset$, we let $\func{n}$ be the identity function. In order to have a reliable benchmark for the variance we first implemented the naive, brute-force estimation technique described in Section~\ref{eq:estimation:asymptotic:variance}, producing 2000 replicates of the particle filter with $N=5000$. Then we computed the sample variance of these filter estimates at each iteration and multiplied the same by $N$. 

Algorithm~\ref{algo:adaptive} was implemented with the two different sample sizes $N=1000$ and $N=100000$ in order to assess stability as well as convergence. The output is displayed in Figures~\ref{fig:compare_1}~and~\ref{fig:compare_2}, where the \namealgo estimator is compared to the brute-force benchmark, the CLE, and the fixed-lag approach of \cite{olsson:douc:2017} with $\lambda \in \{14, 24\}$. In both cases, our estimator produces more precise and stable estimates of the asymptotic variance. Moreover, increasing the number of particles leads to significantly better accuracy, demonstrating the convergence properties of our method. These patterns can also be noticed in Figure~\ref{fig:compare_zoom}, where we focus on large values of $n$. We see that, as expected, that when $n$ is large the CLE either drops to zero or suffers from large variance due to the depletion of the Eve indices. The fixed-lag approach has a similar behavior as our adaptive approach, being both close to the benchmark brute-force values. The fundamental difference is that in the adaptive method the lag is designed adaptively and dynamically, whereas for the fixed-lag method the lag is set to a constant value close to the average lag produced by the \namealgo estimator. We stress again that without access to the \namealgo procedure, the design of a suitable fixed lag $\lambda$ would require an exhaustive prefatory simulation-based analysis, where $\lambda$ is selected by producing multiple fixed-lag variance estimates for a range of different lags and repeated runs of the particle filter and comparing the same to an estimate obtained using the brute-force estimator. 

\begin{figure}[htb]
	\centering
	\includegraphics[width=\linewidth]{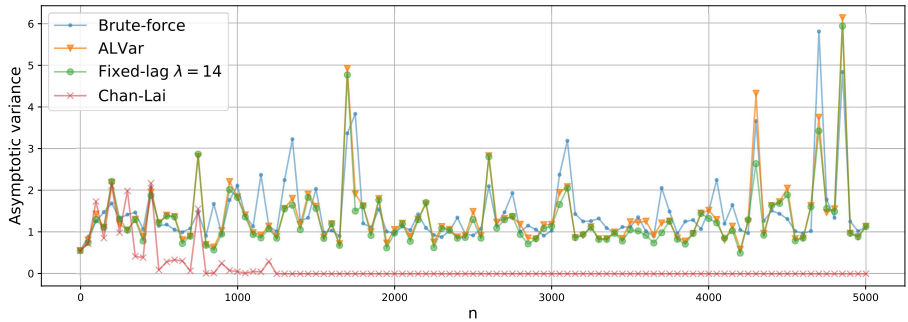}
	\caption{Comparison between variance estimators in the context of optimal filtering in the stochastic volatility model \eqref{eq:stocvol}. The particle sample size was $N=1000$ and the plot displays every 50th variance estimate. The brute-force estimates are based on 2000 replicates of the particle filter.}\label{fig:compare_1}
\end{figure}
\begin{figure}[htb]
	\includegraphics[width=\linewidth]{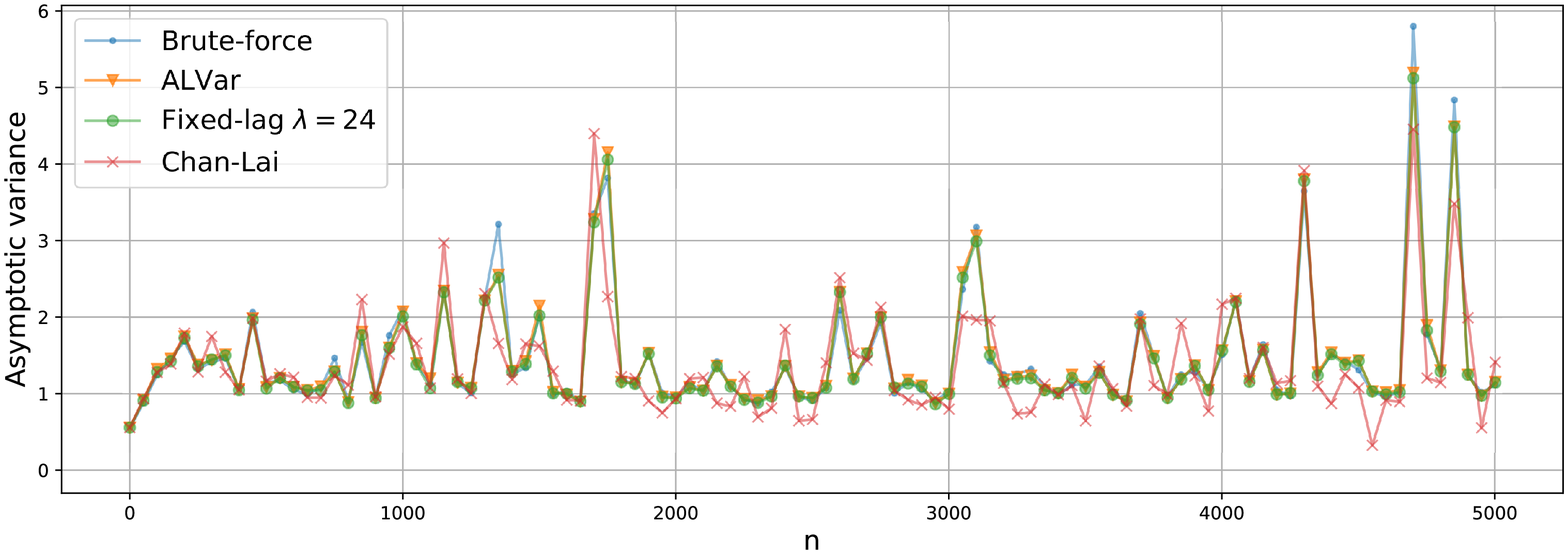}
	\caption{Comparison between variance estimators in the context of optimal filtering in the stochastic volatility model \eqref{eq:stocvol}. The particle sample size was $N=100000$ and the plot displays every 50th variance estimate. The brute-force estimates are based on 2000 replicates of the particle filter.}\label{fig:compare_2}
\end{figure}
\begin{figure}[htb]
	\includegraphics[width=0.49\linewidth]{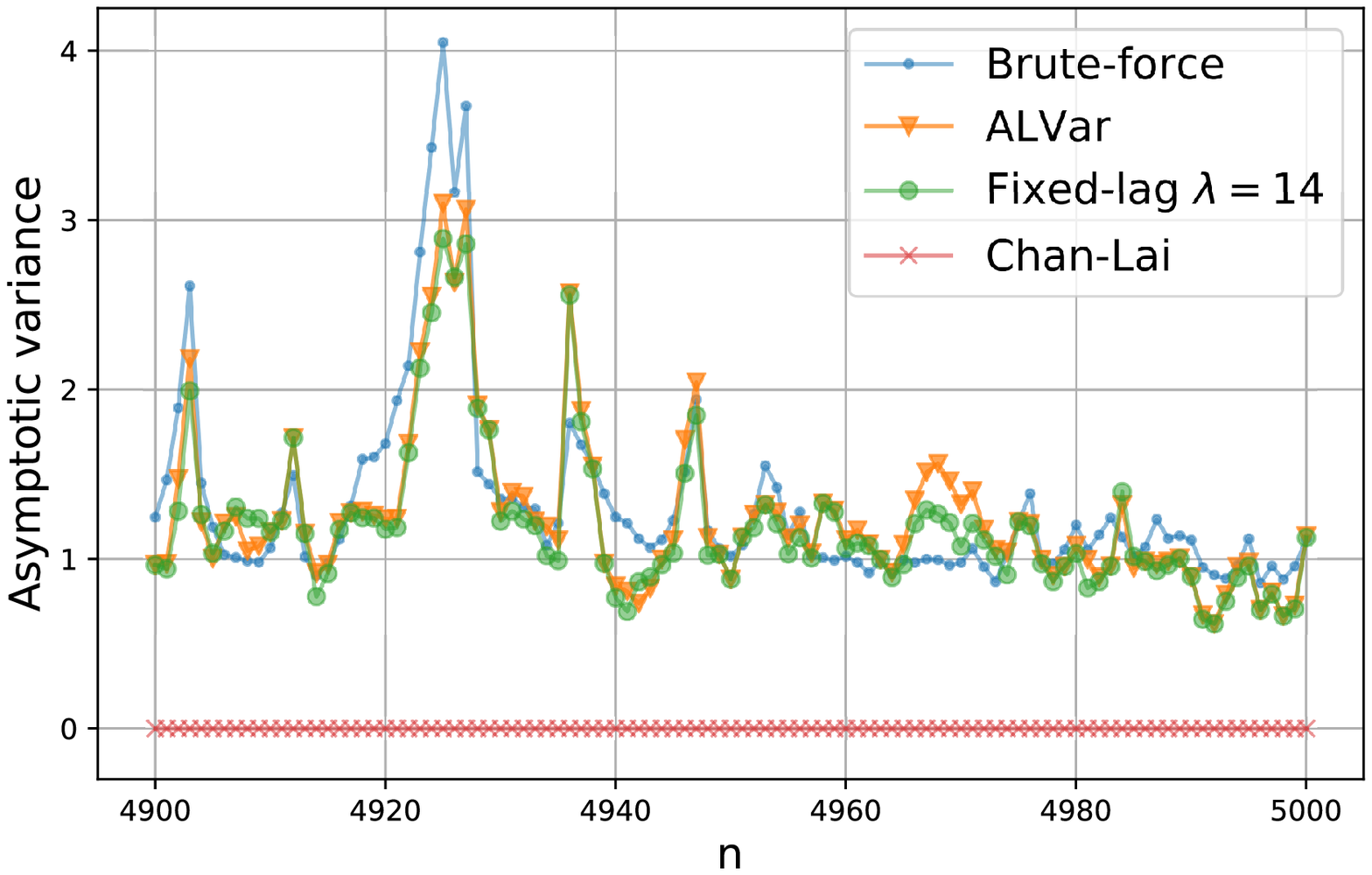}
	\includegraphics[width=0.49\linewidth]{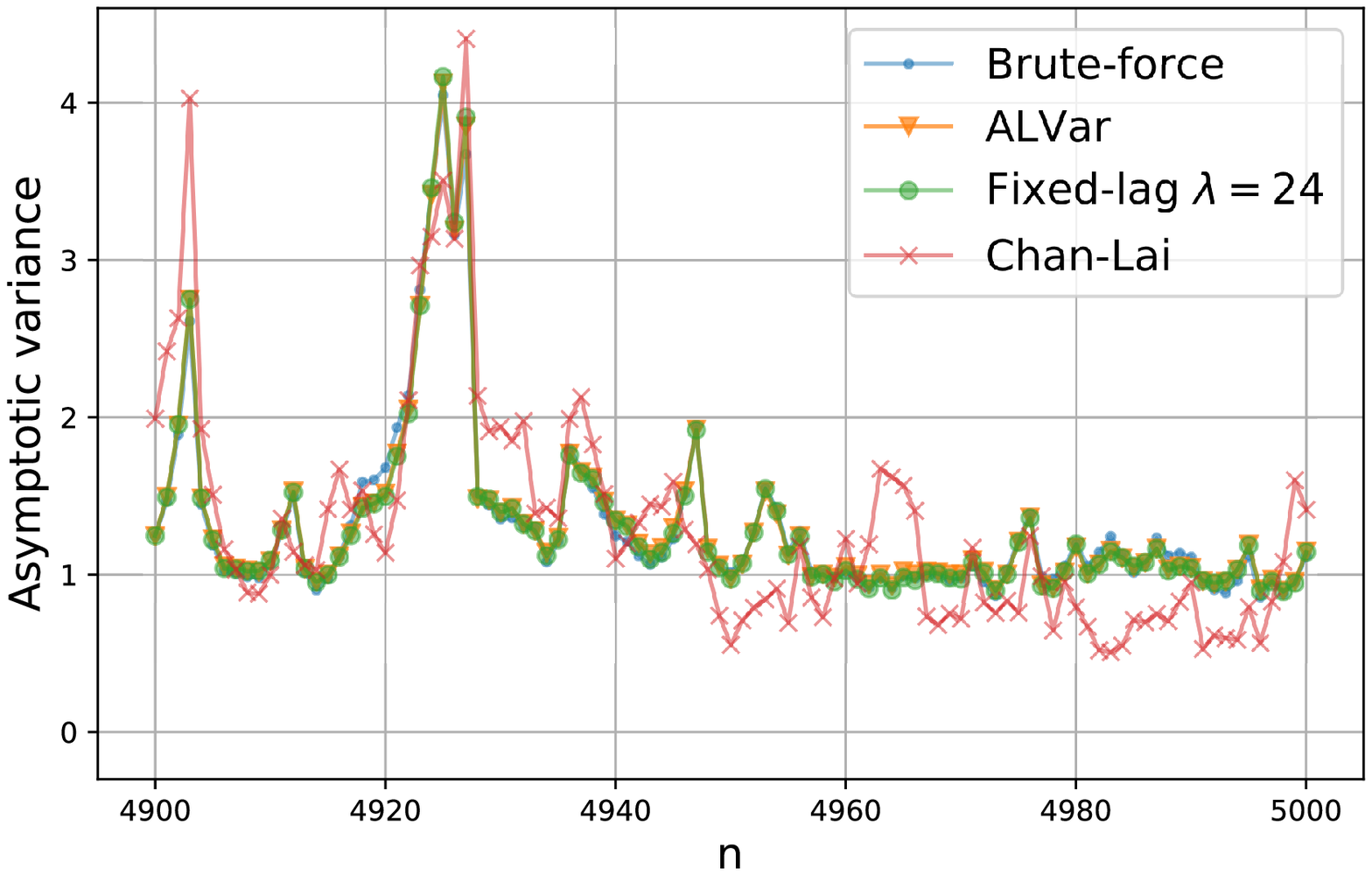}
	\caption{Similar plots as in Figures~\ref{fig:compare_1} and \ref{fig:compare_2}, but with focus on large $n \in \intvect{4900}{5000}$. The left and right panels correspond to $N = 1000$ and $N = 100000$, respectively.}\label{fig:compare_zoom}
\end{figure}

The previous plots are based upon single runs of the algorithms producing variance estimates for different $n\in\intvect{0}{5000}$; we now focus instead on how several estimates are distributed for some specific times $n$, similar to \cite[Figure~2]{olsson:douc:2017}. In the boxplots in Figure~\ref{fig:boxplot}, each box represents the distribution of variance estimates at time $n = 1000$ using the \namealgo algorithm, the CLE, and the fixed-lag approach with several choices of $\lambda$, obtained on the basis of 100 replicates of Algorithm~\ref{algo:pfenoch} for each of these lags. For the box dedicated to the \namealgo we have indicated the average $\lambda_{1000}$ across the 100 independent particle filter replicates (not to be confused with the average lag across all iterations of a single realisation of the particle filter). We observe that our estimator manifests negligible bias, with variability similar to the one of the best fixed-lag estimators.
\begin{figure}[htb]
	\centering
	\includegraphics[width=\linewidth]{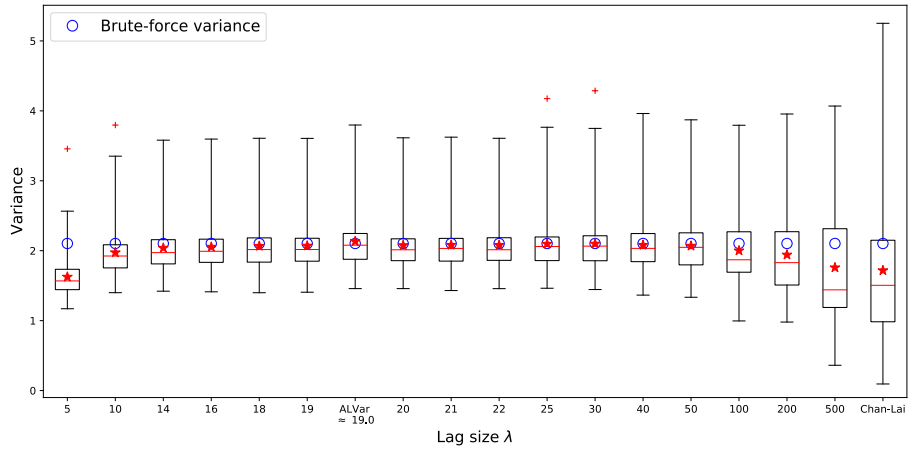}
	\caption{Distribution of fixed-lag and adaptive variance estimates at iteration $n=1000$. Each box is based on 100 replicates of Algorithm~\ref{algo:pfenoch} for a given value of $\lambda$, except the one marked \texttt{\namealgo}, which corresponds to our adaptive method. The average adaptive lag was approximately 19.0. Each particle filter used $N=10000$ particles. The circles correspond to the estimate produced by the brute-force algorithm and the lines and stars in the boxes indicate the medians and means of each sample, respectively.}\label{fig:boxplot}
\end{figure}

\subsubsection{Adaptive-lag analysis }
\label{sec:num:lag:analysis}

In this part we investigate how the values of the lags chosen adaptively at each iteration of Algorithm~\ref{algo:adaptive} are distributed and depend on the number of particles $N$. In Figure~\ref{fig:evolag1000} we show the evolution of the chosen lags over time for $N=1000$ particles. We see that after an initial constant increase, the values stabilise in a range where most of these are between 5 and 30. An interesting pattern is the presence of regimes with constant increase of the lag, during which the same generation of Enoch indices is used, and possibly sudden drops, when the so-far-used generation becomes depleted and a more recent one comes into substitution.
\begin{figure}[H]
	\centering
	\includegraphics[width=0.49\linewidth]{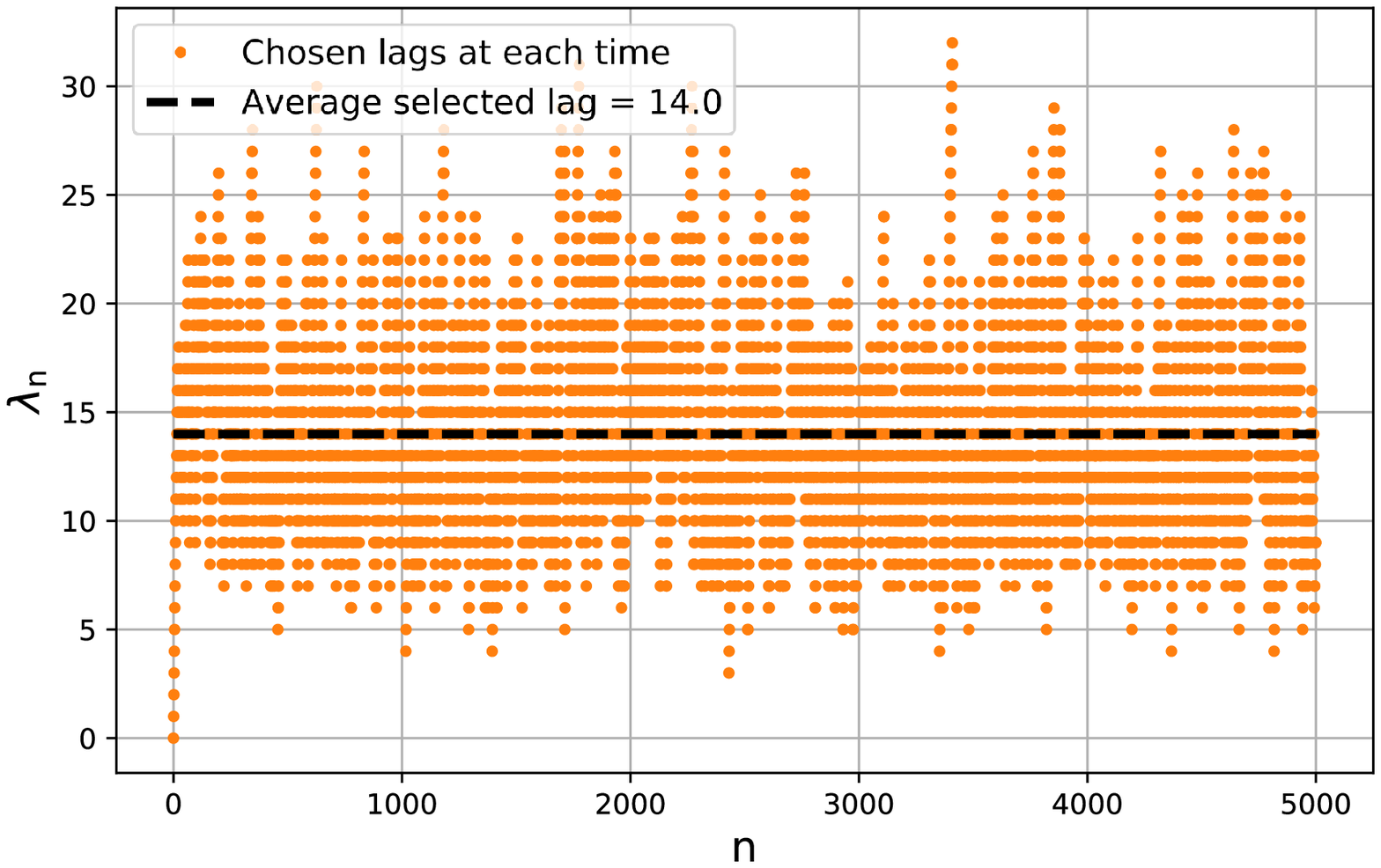}
	\includegraphics[width=0.49\linewidth]{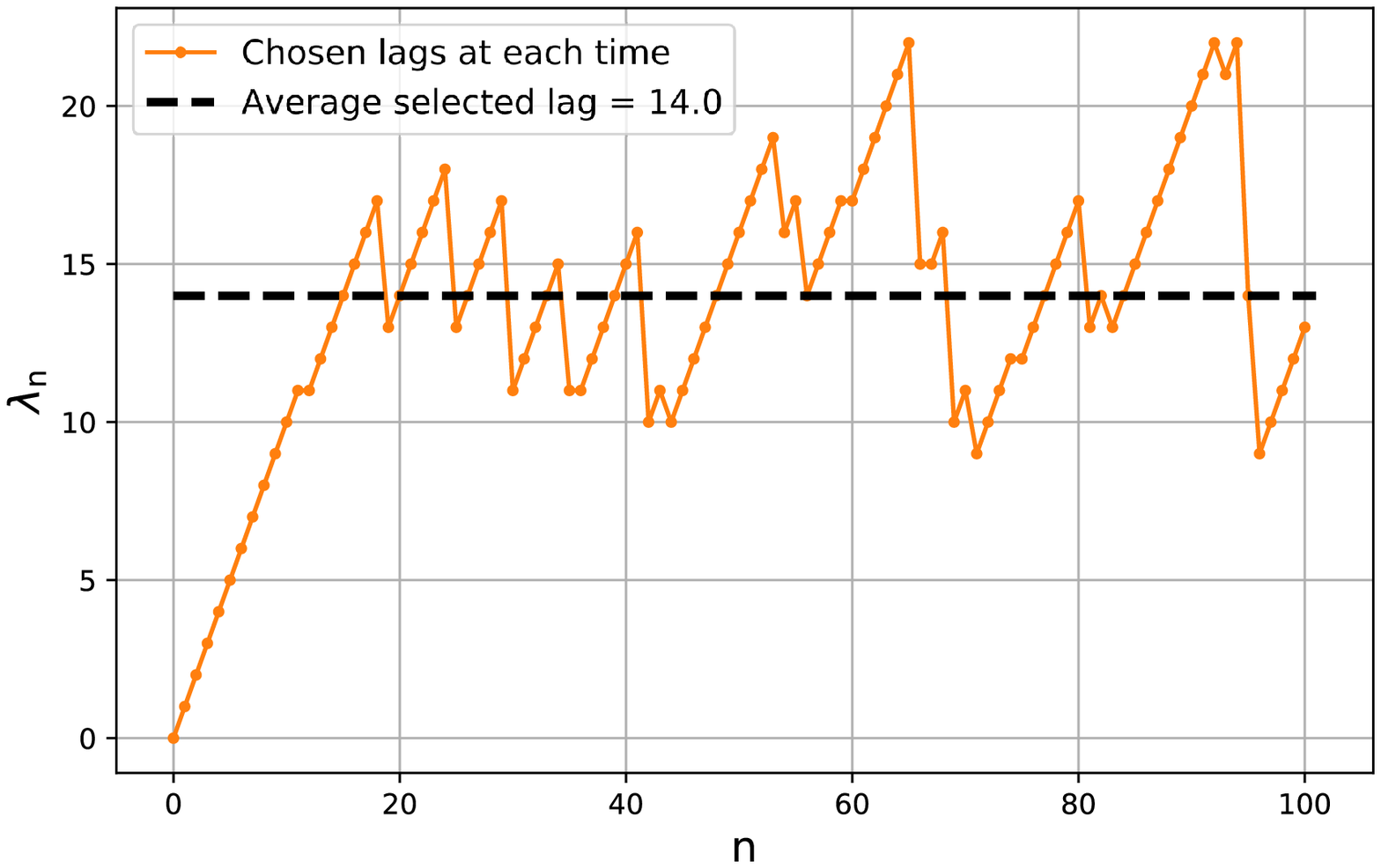}
	\caption{Evolution of the choice of optimal lag for $n \in \intvect{0}{5000}$ (left panel) and $N = 1000$ particles. The right panel shows the 100 first iterations. The average lag was approximately 14.0.}\label{fig:evolag1000}
\end{figure}
We do not show it here, but we notice a similar pattern also when $N$ is increased from $1000$ to $100000$, in which case the range of selected lags is between 10 and 50, with an average around 24. This is not a surprise, since in Theorem~\ref{thm:conv} we have shown that the adaptive lag is expected to converge to the maximum possible value $\lambda_{n}=n$, why we expect larger lags for large sample sizes. The good news is that the complexity does not explode as $N$ is increasing; indeed, the \namealgo algorithm is between 1.5 and 2 times slower than a standard particle filter for $N=1000$ particles and between 2 and 2.5 times slower for $N=100000$ particles. Moreover, our novel method always took significantly less than twice the time of a fixed-lag algorithm with $\lambda$ selected around the average value of the adaptive approach (1.4 and 1.7 times slower for $N=1000$ and $N=100000$ particles, respectively).
The computational time of the \namealgo procedure is closely related to the values of the lags that it chooses across all iterations, since the larger these are, the longer it takes to update the Enoch indices and to make the update on Line~8 in Algorithm~\ref{algo:adaptive}. In Figure~\ref{fig:lag_vs_particles} we illustrate how the lags are distributed for different particle sample sizes; as predicted by our heuristic argument in Section~\ref{subsubsec:loglag}, the dependence of the average lag on $N$ clearly appears to be logarithmic with respect to $N$. Also the maxima behave similarly.  

\begin{figure}[htb]
	\centering
	\includegraphics[width=\linewidth]{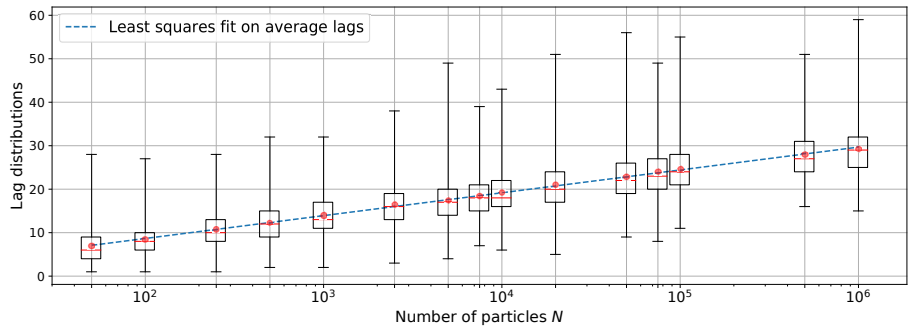}
	\caption{Each box shows the distribution of lags selected adaptively by the \namealgo algorithm in a single run up to iteration $n=5000$ for a given particle sample size $N$ (excluding the first 100 lags in each run). Different boxes correspond to different sample sizes. The dots and lines in each box represent the means and the medians of the sample distributions, respectively. The dashed line is the least squares fit between the average lag of each run and $\log_{10}(N)$.}
	\label{fig:lag_vs_particles}
\end{figure}

\subsubsection{\namealgo vs. empirical MSE}
At the beginning of Section~\ref{sec:results} we claimed that an optimal choice of the lag could be the one minimising the MSE. We now want to check that the lags selected by the \namealgo algorithm are sufficiently close to the ones minimising \eqref{eq:MSE} for most iterations. As we mentioned, \eqref{eq:MSE} is hard to evaluate analytically but can be estimated by means of the \emph{empirical MSE} obtained by running $M\in\nsetpos$ independent particle filters. More precisely, for every $n\in\nset$, $\lambda\in\nset$, and $\func{n}\in\bmf{\alg{X}_n}$, we define the empirical MSE
\begin{equation}
	\emse_n\langle\func{n} \rangle(\lambda) \eqdef \frac{1}{M}\sum_{j=1}^{M}\left( \hat{\sigma}^{2, j}_{n, \lambda}(\func{n}) -\asvar{n}(\func{n})\right)^2,
\end{equation}
at time $n$, where $\hat{\sigma}^{2, j}_{n, \lambda}(\func{n})$ is the estimate produced by the $j$-th particle filter and $\asvar{n}(\func{n})$ can be approximated by a brute-force estimate. Then, for every $n\in\nset$ and $\func{n}\in\bmf{\alg{X}_n}$ we determine the optimal lag by selecting
\begin{equation} \label{eq:MSE:optimal:lag}
	\hat{\lambda}_n^*\gets \argmin_{\lambda\in\intvect{0}{n}}\emse_n\langle\func{n} \rangle(\lambda).
\end{equation}
In order to compare the adaptive lags formed by the \namealgo estimator to the empirical MSE-optimal lags \eqref{eq:MSE:optimal:lag}, we run $M=1000$ particle filters, with each $N=10000$ particles, for $n=500$ iterations; letting $\func{n}=\mathrm{id}$ for all $n$, we determined, for each replicate, the adaptive lags selected by the \namealgo procedure as well as the ones minimising the empirical MSE. Figure~\ref{fig:mse} reports adaptive-lag distributions at some iterations, together with the lag values $\hat{\lambda}_n^*$ minimising the empirical MSE; remarkably, we observe that the empirically optimal lags are within the range of lags selected by the \namealgo algorithm, although the latter tends to choose slightly larger values on average.
\begin{figure}[htb]
	\centering
	\includegraphics[width=\linewidth]{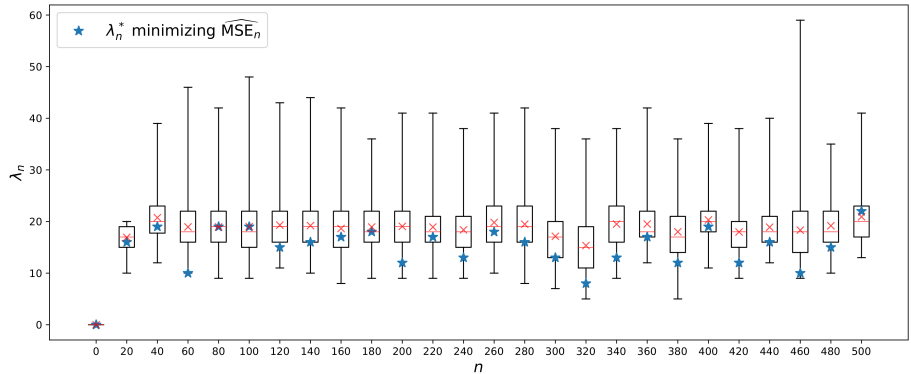}
	\caption{Each box represents the distribution, at a given iteration, of the adaptive lags selected at that time by the \namealgo estimator for the $M=1000$ different replicates. The crosses and the lines in each box are the corresponding means and medians, respectively. The stars are the lags minimising the empirical MSE evaluated at each iteration on the basis of the $M=1000$ replicates.}\label{fig:mse}
\end{figure}

\subsubsection{Variance estimation in the case of adaptive resampling}
In this section we test the \namealgo estimator in the case where the resampling operation is triggered by the ESS criterion according to Algorithm~\ref{algo:adaptivePF}. Figure~\ref{fig:adaptive} displays brute-force estimates of the asymptotic variance as well as estimates produced by the \namealgo estimator in Algorithm~\ref{algo:adaptive2} with two distinct choices of the parameter $\alpha \in \{0,5, 0.2\}$. In both cases, the observations $y_{0:5000}$ generated in the previous section were used as input to the particle filter. Although being based on the same observations and exhibiting similar patterns, we notice that the two brute-force-estimated asymptotic variances differ, as expected, from each other and from the ones reported  Figures~\ref{fig:compare_1}~and~\ref{fig:compare_2}. Still, in both cases the \namealgo estimator is capable of tracking closely the time evolution of the variance. Since the lag is now chosen in terms of selection times and not of simple iterations, the resulting average values are significantly smaller than before, namely 3.0 when $\alpha=0.5$ and $1.9$ when $\alpha=0.2$.
\begin{figure}[htb]
	\centering
	\includegraphics[width=\linewidth]{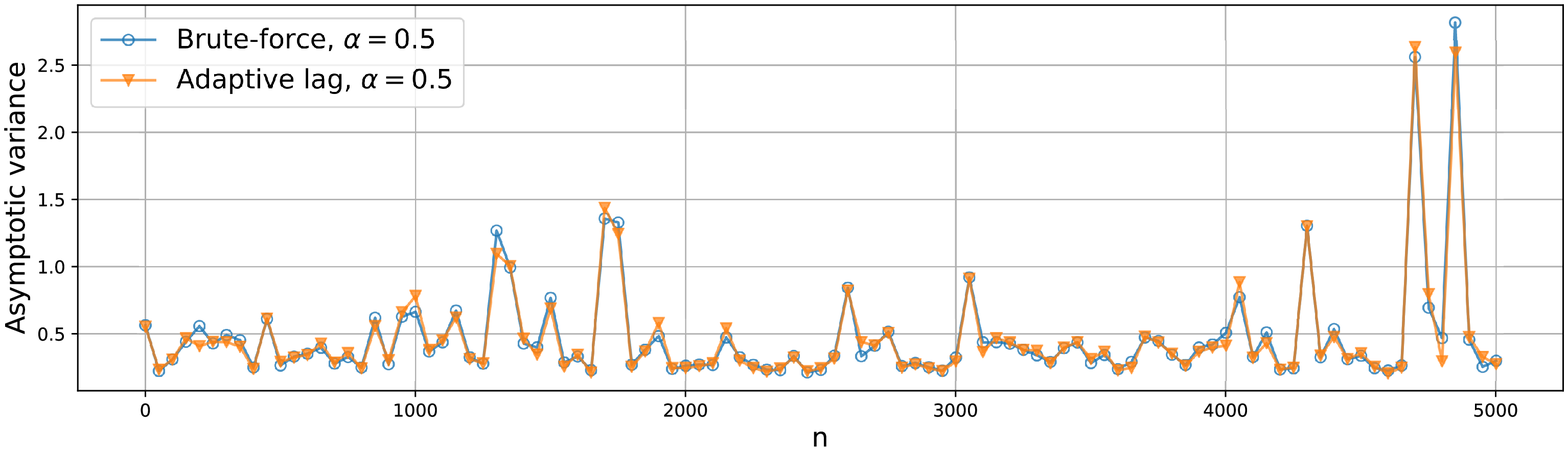}
	\includegraphics[width=\linewidth]{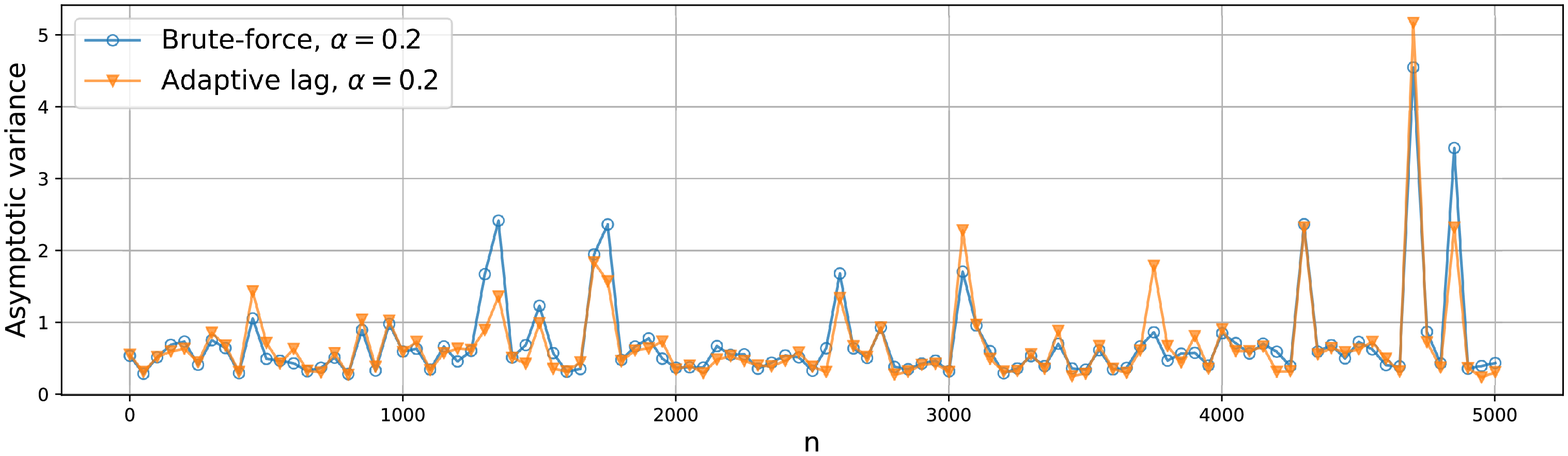}
	\caption{Variance estimation in the APF with ESS-based selection schedule using $N=10000$ particles and a resampling parameter $\alpha$ equal to 0.5 (top panel) and 0.2 (bottom panel), respectively. The plot displays every 50th estimate.}\label{fig:adaptive}
\end{figure}

\subsection{Linear Gaussian SSM and approximate confidence bounds}

In this section we are interested in evaluating the ability of the \namealgo estimator to provide reliable Monte Carlo confidence bounds for the quantities interest. In order to do that, we consider the linear Gaussian SSM
\begin{equation}\label{eq:lingauss}
\begin{split}
	X_{n+1}&=a X_n +\sigma_u U_{n+1}, \\
	Y_n&=b X_n  +\sigma_v V_n, 
\end{split}
\quad n \in \nset,
\end{equation}
where $(U_n)_{n\in \nsetpos}$ and $(V_n)_{n\in\nset}$ are again sequences of uncorrelated standard Gaussian noise variables. We let the parameters be equal to $(a,b,\sigma_u,\sigma_v)=($0.98, 1, 0.2, 1). For linear Gaussian models, the filter distributions are Gaussian and available in a closed form through the \emph{Kalman filter} (see \eg \cite[Section~5.2.3]{cappe:moulines:ryden:2005}), which makes these models particularly well suited for assessing the performance of particle methods. We proceed as follows. Given a sequence  $y_{0:1000}$ of observations, the Kalman filter produces the exact values of $\post{n}(\mathrm{id})=\E[X_n \mid Y_{0:n}=y_{0:n}]$, $n\in\intvect{0}{1000}$; then, for a single run of the \namealgo we may create, for each $n\in\intvect{0}{1000}$ a 95\% confidence interval 
\begin{equation}\label{eq:confint}
\post[N]{n}(\mathrm{id}) \pm \uplambda_{0.025}\frac{\asvarp{n,\lambda_n}(\mathrm{id})}{\sqrt{N}},
\end{equation}
where $\uplambda_{0.025}$ is the 2.5\% quantile of the standard Gaussian distribution. Finally, we let $N=10000$ and produce 200 independent replicates of the APF and associated \namealgo estimator (Algorithm~\ref{algo:adaptive}). For each replicate we calculate, for every $n \in \intvect{0}{1000}$, a 95\% confidence interval \eqref{eq:confint}. Figure~\ref{fig:failrate1} reports the failure rate, \ie, the ratio of cases in which the true value falls outside the corresponding confidence interval, at each iteration. We observe an average failure rate across all times around $5.2\%$ (instead of the ideal $5\%$ failure rate), which may suggest a slight underestimation of the asymptotic variance.
\begin{figure}[htb]
	\centering
	\includegraphics[width=\linewidth]{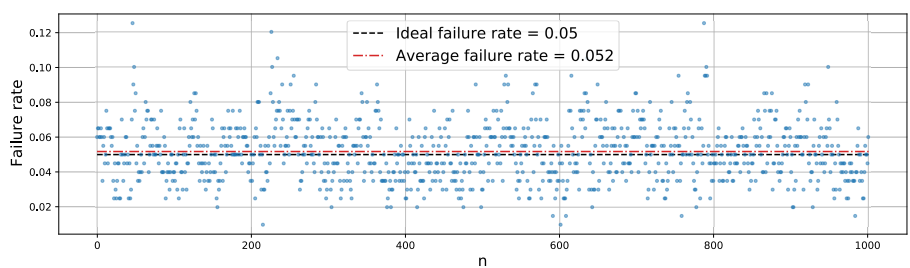}
	\caption{Empirical failure rates across iterations $n\in\intvect{0}{1000}$. Each rate is obtained on the basis of 200 replicates of Algorithm~\ref{algo:adaptive} with $N=10000$. The dashed line represents the perfect rate of $5\%$, while the dash-dotted line is the overall average failure rate.}\label{fig:failrate1}
\end{figure}
Similar results was obtained with the ESS-based approach described by Algorithm~\ref{algo:adaptive2}, where using $\alpha = 0.2$ and $\alpha = 0.5$ lead to average failure rates of $5.2\%$ and $5.0\%$, respectively. All in all, these results confirm that our approach is reliable and has small or negligible bias.

%% file: Appendix.tex
\section{Proof of Theorem~\ref{thm:consistency-fixed-lag}}\label{sec:appendix1}

In the following we will prove Theorem~\ref{thm:consistency-fixed-lag} by extending \cite[Proposition~8]{olsson:douc:2017}, which provides the convergence of interest in the special case of Feynman--Kac models and standard bootstrap filters, to the general model context in Section~\ref{subsec:model} and APFs. For this purpose, we first need to reformulate the APF introduced Section~\ref{subsec:smc} as a simple bootstrap filter operating on an auxiliary, extended Feynman--Kac model. More precisely, define a sequence  $(\Xbar{n},\Xalgbar{n})_{n\in\nset}$ of measurable spaces by setting $\Xbar{0} \eqdef \set{X}_0$ and $\Xalgbar{0} \eqdef \alg{X}_0$ and, for $n \in \nsetpos$, $\Xbar{n} \eqdef \set{X}_{n-1} \times \set{X}_n$ and $\Xalgbar{n} \eqdef \alg{X}_{n-1} \varotimes \alg{X}_n$. Moreover, elements in these spaces will be denoted by $\xbar{0} \eqdef x_0$, $\xbar{n} \eqdef (x_{n-1}, x_n)$ and we will also write $\xbar{n}^1 \eqdef x_{n - 1}$ and $\xbar{n}^2 \eqdef x_n$ to indicate the first and second element of $\xbar{n}$, respectively. Let $\bar{\Xinitprop} \eqdef \Xinitprop$ and define $\gbar{0}(\xbar{0}) \eqdef \am{0}(x_0)\rn{-1}(x_0)$. We also define the initial unnormalised measure $\bar{\Xinit} \in \meas{\Xalgbar{0}}$ such that for $A \in \Xalgbar{0}$, $\bar{\Xinit}(A) = \bar{\Xinitprop}(\gbar{0} \1{A})$ (notice that $\bar{\Xinit}\ne\Xinit$). Now, define, for each $n\in\nsetpos$, the auxiliary potential function
\begin{equation}\label{eq:gbar}
\gbar{n}(\xbar{n}) \eqdef \rn{n-1}(x_{n-1}, x_{n}) \frac{\am{n}(x_{n})}{\am{n-1}(x_{n-1})},\quad \xbar{n}\in\Xbar{n}.
\end{equation}
and the Markov transition kernel 
\begin{equation}\label{eq:hkbar}
\hkbar{n}(\xbar{n}, d\xbar{n+1}) \eqdef \delta_{\xbar{n}^2}(d\xbar{n+1}^1) \, \propker{n}(\xbar{n+1}^1, d\xbar{n+1}^2), \quad \xbar{n}\in\Xbar{n}.
\end{equation}
Thus, as established by Lemma~\eqref{lemma:apf_bpf_dist} below, Algorithm~\ref{algo:apf} may now be reinterpreted as a bootstrap particle filter operating on the auxiliary Feynman--Kac model given by \eqref{eq:gbar} and \eqref{eq:hkbar}. Algorithm~\ref{algo:bpf} shows one iteration of this procedure, which is initialised by sampling $(\epartbar{0}{i})_{i=1}^N$ from $\bar{\Xinitprop}^{\varotimes N}=\Xinitprop^{\varotimes N}$ and letting $\wgtbar{0}{i}\gets\gbar{0}(\epartbar{0}{i})$ for all $i$.
\begin{algorithm}[htb]
	\caption{Bootstrap particle filter on auxiliary Feynman--Kac model}\label{algo:bpf}
	\begin{algorithmic}[1]
		\Statex Input: $(\epartbar{n}{i},\wgtbar{n}{i})_{i=1}^\N$.
		\Statex Output $(\epartbar{n+1}{i},\wgtbar{n+1}{i},\Ibar{n+1}{i})_{i=1}^\N$.
		\For{$i=1\rightarrow\N$}
		\State draw $\Ibar{n+1}{i}\sim \catdist((\wgtbar{n}{\ell})_{\ell=1}^\N)$;
		\State draw $\epartbar{n+1}{i}\sim \hkbar{n}(\epartbar{n}{\Ibar{n+1}{i}},\cdot)$;
		\State set $\wgtbar{n+1}{i}\leftarrow  \gbar{n+1}(\epartbar{n+1}{i})$;
		\EndFor
	\end{algorithmic}
\end{algorithm}
\begin{lemma}\label{lemma:apf_bpf_dist}
	Let $(\epart{n}{i},\wgt{n}{i},\I{n}{i})_{i=1}^\N$ and $(\epartbar{n}{i},\wgtbar{n}{i},\Ibar{n}{i})_{i=1}^\N$ be the particles, weights, and resampling indices produced by $n$ iterations of Algorithms~\ref{algo:apf} and \ref{algo:bpf}, respectively.
	Then for every $n\in\nsetpos$,  
	\begin{equation}
	(\epart{n}{i},\wgt{n}{i}\am{n}(\epart{n}{i}),\I{n}{i})_{i=1}^\N\eqdist(\epartbar{n}{i,2},\wgtbar{n}{i},\Ibar{n}{i})_{i=1}^\N,
	\end{equation}
	where $\epartbar{n}{i,2}$ indicates the second component of $\epartbar{n}{i}\eqdef(\epartbar{n}{i,1},\epartbar{n}{i,2})$.
\end{lemma}
\begin{proof}
	First note that by construction, $(\epart{0}{i},\wgt{0}{i}\am{0}(\epart{0}{i}))_{i=1}^\N\eqdist(\epartbar{0}{i},\wgtbar{0}{i})_{i=1}^\N$. We may hence proceed by induction and assume that $(\epart{n}{i},\wgt{n}{i}\am{n}(\epart{n}{i}))_{i=1}^\N\eqdist(\epartbar{n}{i,2},\wgtbar{n}{i})_{i=1}^\N$ for some $n\in\nset$ (in the case $n=0$, we denote $\epartbar{0}{i,2} \eqdef \epartbar{0}{i}$). In Algorithm~\ref{algo:apf} we draw the conditionally i.i.d. resampling indices $(\I{n+1}{i})_{i=1}^N$ according to
	\begin{equation}
	\I{n+1}{i} \sim \catdist((\wgt{n}{\ell}\am{n}(\epart{n}{\ell}))_{\ell=1}^\N),\quad i\in\intvect{1}{N},
	\end{equation}
	whereas in Algorithm~\ref{algo:bpf} we draw $(\Ibar{n+1}{i})_{i=1}^N$ according to
	\begin{equation}
	\Ibar{n+1}{i} \sim \catdist((\wgtbar{n}{\ell})_{\ell=1}^\N) \eqdist \catdist((\wgt{n}{\ell}\am{n}(\epart{n}{\ell}))_{\ell=1}^\N),\quad i\in\intvect{1}{N},
	\end{equation}
	which implies that $(\I{n+1}{i})_{i=1}^N\eqdist(\Ibar{n+1}{i})_{i=1}^N$. Furthermore, in Algorithm~\ref{algo:apf} the particles are propagated by sampling, for $i \in\intvect{1}{N}$, 
	\begin{equation}
	\epart{n+1}{i}\sim \propker{n}(\epart{n}{\I{n+1}{i}}, dx_{n+1}),
	\end{equation}
	whereas in Algorithm~\ref{algo:bpf}, for $i \in\intvect{1}{N}$, 
	\begin{equation}
	\epartbar{n+1}{i}\sim \hkbar{n}(\epartbar{n}{\Ibar{n+1}{i}},d\xbar{n+1})
	\end{equation}
	which corresponds to assigning
	\begin{equation}
	\epartbar{n+1}{i,1}\gets \epartbar{n}{\Ibar{n+1}{i},2} 
	\end{equation}
	and then sampling 
	\begin{equation}
	\epartbar{n+1}{i,2}\sim \propker{n}(\epartbar{n}{\Ibar{n+1}{i},2}, d\xbar{n+1}^2).
	\end{equation}
	Using the induction hypothesis and the previous result on the resampling indices, this implies that $(\epart{n+1}{i})_{i=1}^\N\eqdist(\epartbar{n+1}{i,2})_{i=1}^\N$. Finally, since the weights are functions only of the particles and the indices, it holds, for $i\in\intvect{1}{N}$,
	\begin{multline}
	\wgt{n+1}{i}\am{n+1}(\epart{n+1}{i})= \rn{n}(\epart{n}{\I{n+1}{i}}, \epart{n+1}{i}) \frac{\am{n+1}(\epart{n+1}{i})}{\am{n}(\epart{n}{\I{n+1}{i}})}
	\\
	\eqdist\rn{n}(\epartbar{n}{\Ibar{n+1}{i},2}, \epartbar{n+1}{i,2}) \frac{\am{n+1}(\epartbar{n+1}{i,2})}{\am{n}(\epartbar{n}{\Ibar{n+1}{i},2})}=\rn{n}(\epartbar{n+1}{i,1}, \epartbar{n+1}{i,2})\frac{\am{n+1}(\epartbar{n+1}{i,2})}{\am{n}(\epartbar{n+1}{i,1})}
	\\
	=\gbar{n+1}(\epartbar{n+1}{i})=\wgtbar{n+1}{i}, 
	\end{multline}
	which shows that $(\epart{n+1}{i},\wgt{n+1}{i}\am{n+1}(\epart{n+1}{i}),\I{n+1}{i})_{i=1}^\N\eqdist(\epartbar{n+1}{i,2},\wgtbar{n+1}{i},\Ibar{n+1}{i})_{i=1}^\N$. This completes the proof. 
\end{proof}
We are now ready to prove Theorem~\ref{thm:consistency-fixed-lag}.
\begin{proof}[Proof of Theorem~\ref{thm:consistency-fixed-lag}]
	Consider a given iteration index $n\in\nset$; if we are interested in estimating the variance of the particle estimator after $n$ iterations, we may assume that $\am{n}\equiv 1$, since the adjustment multiplier $\am{n}$ only influences the distribution of the particles of the APF at the next iterations $n + 1, n + 2, \ldots$ Now, the estimate of $\post{n} \func{n}$ produced by Algorithm~\ref{algo:apf} for a given $\func{n}\in\bmf{\alg{X}_n}$ can be interpreted as a statistically equivalent estimate formed by Algorithm~\ref{algo:bpf} by defining $\funcbar{n} : \Xbar{n} \ni \xbar{n} \mapsto \func{n}(\xbar{n}^2)$. Then by Lemma~\ref{lemma:apf_bpf_dist},
	\begin{equation}
	\postbar[N]{n}\funcbar{n} \eqdef \sum_{i=1}^{N}\frac{\wgtbar{n}{i}}{\wgtsumbar{n}}\funcbar{n}(\epartbar{n}{i})\eqdist\sum_{i=1}^{N}\frac{\wgt{n}{i}}{\wgtsum{n}}\func{n}(\epart{n}{i})=\post[N]{n}\func{n}.
	\end{equation}
	Next, we also notice that the expectations $\postbar{n}\funcbar{n}$ and $\post{n}\func{n}$ coincide; indeed, for every $m\in\nsetpos$, we may define 
	\begin{multline}
	\predkerbar{\ell}{m-1} \funcbar{m}(\xbar{\ell}) \eqdef \idotsint \funcbar{m}(\xbar{m}) \prod_{k = \ell}^{m - 1}  \gbar{k+1}(\xbar{k + 1}) \, \hkbar{k}(\xbar{k}, d \xbar{k + 1}), \\\quad \xbar{\ell} \in \Xbar{\ell}, \funcbar{m} \in \bmf{\Xbar{m}}, 
	\end{multline}
	for $\ell \in \intvect{1}{n - 1}$ and $ \predkerbar{\ell}{m-1}=\mathrm{id} $ otherwise; 
	moreover, for any function $\funcbar{m} \in \bmf{\Xalgbar{m}}$ such that for all $\xbar{m}\in\Xbar{m}$, $\funcbar{m}(\xbar{m})=\func{m}(\xbar{m}^2)$ for some $\func{m}\in\bmf{\alg{X}_{m}}$ it holds, for every $\ell \in \intvect{0}{m - 1}$,
	\begin{align*}
	\lefteqn{\predkerbar{\ell}{m-1}\funcbar{m}(\xbar{\ell})} 
	\\
	&=\int\cdots\int \func{m}(\xbar{m}^2)\prod_{k=\ell}^{m-1}\rn{k}(\xbar{k+1}^1, \xbar{k+1}^2)\frac{\am{k+1}(\xbar{k+1}^2)}{\am{k}(\xbar{k+1}^1)}\delta_{\xbar{k}^2}(d\xbar{k+1}^1)\,\propker{k}(\xbar{k+1}^1,d\xbar{k+1}^2)
	\\
	&=\int\cdots\int \func{m}(x_{m})\prod_{k=\ell}^{m-1}\rn{k}(x_k, x_{k+1})\frac{\am{k+1}(x_{k+1})}{\am{k}(x_k)}\,\propker{k}(x_k,dx_{k+1})
	\\
	&=\frac{1}{\am{\ell}(x_\ell)}\lk{\ell:m-1}(\am{m}\func{m})(x_\ell).
	\end{align*}
	This implies that   
	\begin{multline}
	\postbar{m}\funcbar{m} = \frac{\bar{\Xinit}\predkerbar{0}{m-1}\funcbar{m}}{\bar{\Xinit}\predkerbar{0}{m-1}\1{\Xbar{m}}}
	=\frac{\Xinitprop(\rn{-1}\am{0}\frac{1}{\am{0}}\lk{0:m-1}(\am{m}\func{m}))}{\Xinitprop(\rn{-1}\am{0}\frac{1}{\am{0}}\lk{0:m-1}\am{m})}
	\\=\frac{\Xinit\lk{0:m-1}(\am{m}\func{m})}{\Xinit\lk{0:m-1}\am{m}}=\frac{\post{m}(\am{m}\func{m})}{\post{m}\am{m}}.
	\end{multline} 
	Thus, since we are assuming that $\am{n}\equiv1$, it follows that $\postbar{n}\funcbar{n} = \post{n}\func{n}$. Applying the CLT in Proposition~\ref{prop:clt} to the auxiliary particle model yields, since $(\gbar{n})_{n\in\nset}$ are all bounded by Assumption~\ref{assum:bounded}, as $N\to \infty$,
	\begin{equation}
	\sqrt{N}(\postbar[N]{n}\funcbar{n}-\postbar{n}\funcbar{n})\convd \bar{\sigma}_n(\funcbar{n})Z,
	\end{equation}
	for some asymptotic variance $\bar{\sigma}_n^2(\funcbar{n})$. However, since $\postbar[N]{n}\funcbar{n}\eqdist\post[N]{n}\func{n}$ and $\postbar{n}\funcbar{n}=\post{n}\func{n}$, it necessarily holds that $\bar{\sigma}_n^2(\funcbar{n})=\sigma_n^2(\func{n})$, where $\sigma_n^2(\funcbar{n})$ is given by \eqref{eq:asvar} with $\ell=0$. It is also easily checked that the equality holds for all the truncated variances. This is done by first  recalling that Algorithm~\ref{algo:bpf} may be seen as a particular case of the general framework under consideration, for which the Radon--Nikodym derivatives with respect to the initial proposal density and the proposal Markov kernels are given by the potential functions and, moreover, the adjustment multipliers are all assumed to be equal to one. Thus, we may write, for all $\ell\in\intvect{0}{n}$,
	\begin{align}
	\lefteqn{\bar{\sigma}_{\ell,n}^2(\funcbar{n})} \\
	&=\frac{\bar{\Xinit}(\gbar{0}\{\predkerbar{0}{n-1}(\funcbar{n}-\postbar{n} \funcbar{n})\}^2)}{(\bar{\Xinit}\predkerbar{0}{n-1}\1{\Xbar{n}})^2}\1{\{\ell=0\}}
	\\&\quad+\sum_{m=(\ell-1)\vee 0}^{n-1}\frac{\postbar{m}\bar{\mathbf{L}}_m(\gbar{m+1}\{\predkerbar{m+1}{n-1}(\funcbar{n}-\postbar{n} \funcbar{n})\}^2)}{(\postbar{m}\predkerbar{m}{n-1}\1{\Xbar{n}})^2}
	\\&=\frac{\bar{\Xinitprop}(\gbar{0}^2\{\am{0}^{-1}\lk{0:n-1}(\func{n}-\post{n} \func{n})\}^2)}{\{\bar{\Xinitprop}(\gbar{0}\am{0}^{-1}\lk{0:n-1}\1{\set{X}_n})\}^2}\1{\{\ell=0\}}
	\\&\quad+\sum_{m=(\ell-1)\vee 0}^{n-1}\post{m}\am{m}\frac{\post{m}\{\am{m}\am{m}^{-1}\lk{m}(\am{m+1}\gbar{m+1}\{\am{m+1}^{-1}\lk{m+1:n-1}(\func{n}-\post{n} \func{n})\}^2)\}}{\{\post{m}(\am{m}\am{m}^{-1}\lk{m:n-1}\1{\set{X}_{n}})\}^2}
	\\&=\frac{\Xinit(\rn{-1}\{\lk{0:n-1}(\func{n}-\post{n} \func{n})\}^2)}{(\Xinit\lk{0:n-1}\1{\set{X}_n})^2}\1{\{\ell=0\}}
	\\&\quad+\sum_{m=(\ell-1)\vee 0}^{n-1}\post{m}\am{m}\frac{\post{m}\lk{m}(\am{m}^{-1}\rn{m}\{\lk{m+1:n-1}(\func{n}-\post{n} \func{n})\}^2)}{(\post{m}\lk{m:n-1}\1{\set{X}_{n}})^2}, 
	\end{align}
	and by \eqref{eq:asvar} we may conclude that 
	\begin{equation}
	\bar{\sigma}_{\ell,n}^2(\funcbar{n}) = \asvar{\ell,n}(\func{n}). \label{eq:compare-truncvar}
	\end{equation}
	Finally, recall that for every $m\in\intvect{0}{n}$, the Enoch indices $(\enochbar{m,n}{i})_{i=1}^N$ are, for all $k\in\intvect{0}{n}$ and $i\in\intvect{1}{N}$, recursively defined by 
	\begin{equation}
		\enochbar{m,k}{i} \eqdef
		\begin{cases}
		i\quad &\text{for }m=k,\\\enochbar{m,k-1}{\Ibar{k}{i}}&\text{for }m<k.
		\end{cases}
	\end{equation}
	Consequently, these indices are functions of the resampling indices, and by Lemma~\ref{lemma:apf_bpf_dist} it holds that $(\enochbar{m,n}{i})_{i=1}^N\eqdist (\enoch{m,n}{i})_{i=1}^N$ (where the latter are the Enoch indices of the particle system generated by Algorithm~\ref{algo:pfenoch}). Now, if we let, for any $\lambda\in \nset$, $\hat{\bar{\sigma}}_{n,\lambda}^2(\funcbar{n})$ be the lag-based variance estimator for the auxiliary bootstrap particle model, we have, again as a consequence of Lemma~\ref{lemma:apf_bpf_dist}, recalling that $\funcbar{n}(\xbar{n}) = \func{n}(\xbar{n}^2)$ and $\am{n}\equiv 1$,
	\begin{align}
	\hat{\bar{\sigma}}_{n,\lambda}^2(\funcbar{n}) &\eqdef N\sum_{i=1}^N\left(\sum_{j:\enochbar{n\langle\lambda\rangle,n}{j}=i}\frac{\wgtbar{n}{j}}{\wgtsumbar{n}}\{\funcbar{n}(\epartbar{n}{j})-\postbar[N]{n}\funcbar{n}\}\right)^2 \\
	\\
	&\eqdist N\sum_{i=1}^N\left(\sum_{j:\enoch{n\langle\lambda\rangle,n}{j}=i}\frac{\wgt{n}{j}}{\wgtsum{n}}\left\{\func{n}(\epart{n}{j})-\post[N]{n}\func{n}\right\}\right)^2, 
	\end{align}
	and by \eqref{eq:OD} we may thus conclude that 
	\begin{equation} \label{eq:compare-estim}
	\hat{\bar{\sigma}}_{n,\lambda}^2(\funcbar{n}) \eqdist \asvarp{n,\lambda}(\func{n}). 
	\end{equation}
	Finally, note that the assumptions of the theorem imply that for all $n\in\nset$ and $\xbar{n}\in\Xbar{n}$, $\gbar{n}(\xbar{n})>0$ and $\sup_{\xbar{n}\in\Xbar{n}}\gbar{n}(\xbar{n})<\infty$, which means that the assumptions of \cite[Proposition~8]{olsson:douc:2017} are satisfied as well. This implies that $\hat{\bar{\sigma}}_{n,\lambda}^2(\funcbar{n})\convp \bar{\sigma}_{n\langle\lambda\rangle,n}^2(\funcbar{n})$ as $N \to \infty$, and combining this result with \eqref{eq:compare-truncvar} and \eqref{eq:compare-estim} yields immediately that $ \asvarp{n,\lambda}(\func{n})\convp\asvar{n\langle\lambda\rangle,n}(\func{n})$ as $N \to \infty$. The proof is complete. 
\end{proof}

\section{Extensions to APFs with adaptive resampling}\label{sec:appendix2}
\subsection{Model extension}\label{subsec:modelext}
The purpose of this section is to show, using arguments similar to those in  \cite[Appendix~C.1]{mastrototaro:olsson:alenlov:2021}, that the results for the standard bootstrap particle filter can be used directly to establish the asymptotic normality in the case where the selection schedule is deterministic and determined by to some given sequence $(\res{n})_{n\in\nset}$ of indicators, where $\res{n} = 1$ means that resampling happens at time $n$. To do this, we will extend the model of Section~\ref{subsec:model} by allowing the states to be represented by paths different lengths determined by $(\res{n})_{n\in\nset}$. First, define the resampling times $n_m \eqdef \min\{n \in \nset : \nres{n+1} = m + 1\}$, $m \in \nset$, where $(\nres{n})_{n \in \nset}$ are defined in Section~\ref{sec:extension:adaptive:resampling}, and, by convention, let $n_{-1} \eqdef -1$. Then we introduce the sequence $(\Xp{m}, \Xpalg{m})_{m \in \nset}$ of measurable spaces, where $\Xp{m} \eqdef \set{X}_{n_{m-1}+1 }\times \set{X}_{n_{m-1}+2} \times \cdots \times \set{X}_{n_{m}}$ and $\Xpalg{m} \eqdef \alg{X}_{n_{m-1}+1}\varotimes \alg{X}_{n_{m-1}+2}\varotimes \cdots\varotimes \alg{X}_{n_m}$. As a general rule, we use boldface to indicate that a quantity is related to such an extended path space; \eg, we let $\xvec{m} \eqdef x_{n_{m-1}+1:n_m}$ indicate a generic element in $\Xp{m}$ and define, for every $m \in \nsetpos$ and $k\in\intvect{n_{m-1}+1}{n_m}$, the projection 
\begin{equation}
	\proj_k^m : \Xp{m} \ni \xvec{m} \mapsto x_{k} \in \set{X}_{k}. 
\end{equation}
We then introduce the multi-step unnormalised transition kernels $(\lkm{m})_{m \in \nset}$, obtained by tensor-products of the single-step Markov transition kernels $(\lk{n})_{n\in\nset}$ of Section~\ref{subsec:model}; more precisely, for every $m \in \nset$, 
\begin{multline}
	\lkm{m} \boldsymbol{h}(\xvec{m}) \eqdef \lk{n_m} \varotimes \lk{n_m+1} \varotimes \cdots \varotimes \lk{n_{m+1}-1} \boldsymbol{h} (\proj_{n_m}^m(\xvec{m})), \\
	\quad (\xvec{m}, \boldsymbol{h}) \in \Xp{m} \times \bmf{\Xpalg{m + 1}}. 
\end{multline} 
Note that $\lkm{m}$ depends only on $x_{n_m}$ and is constant with respect to the previous states. Moreover, note that if selection is performed at each iteration, then $n_m =m$ for all $m$, implying $\lkm{m} = \lk{m}$. The initial measure $\Xinitm$ on $\Xpalg{0}$ is defined as $\Xinitm \eqdef \Xinit \varotimes \lk{0} \varotimes \cdots \varotimes \lk{n_0-1}$. Again, for compactness, we write $\lkm{k:\ell} \eqdef \lkm{k}\cdots\lkm{\ell}$ whenever $k \in \intvect{0}{\ell}$, otherwise $ \lkm{k:\ell}=\mathrm{id} $. Next, for every $n\in\nset$, we define the distribution flow 
$$
\postm{m}:\bmf{\Xpalg{m}}\ni \funcm{m}\mapsto \frac{\Xinitm\lkm{0:m-1}\funcm{m}}{\Xinitm\lkm{0:m-1}\1{\Xp{m}}}.
$$
In order to apply an APF to this model we introduce auxiliary functions $(\amm{m})_{m\in\nset}$ defined by 
$$
\amm{m}(\xvec{m}) \eqdef \am{n_m}(\proj_{n_m}^m(\xvec{m}))=\am{n_m}(x_{n_m}), \quad \xvec{m}\in\Xp{m}. 
$$
After resampling, the particles are propagated according to the Markov proposal kernels $(\propkerm{m})_{m \in \nset}$, where 
\begin{multline*}
\propkerm{m} \boldsymbol{h}(\xvec{m}) \eqdef \propker{n_m} \varotimes \propker{n_m+1} \varotimes \cdots \varotimes \propker{n_{m+1}-1} \boldsymbol{h} (\proj_{n_m}^m(\xvec{m})), \\ 
\quad (\xvec{m}, \boldsymbol{h}) \in \Xp{m} \times \bmf{\Xpalg{m + 1}}. 
\end{multline*} 
Under the assumptions of the paper, each probability measure $\propkerm{n}(\xvec{m},\cdot)$, $\xvec{m}\in\Xp{m}$, is absolutely continuous with respect to $\lkm{m}(\xvec{m},\cdot)$. Hence, for every $\xvec{m}$, we may let $\rnm{m}(\xvec{m}, \xvec{m + 1}) = d \lkm{m}(\xvec{m},\cdot) / d \propkerm{n}(\xvec{m},\cdot)$, $\xvec{m + 1} \in \Xp{m + 1}$, be the corresponding Radon--Nikodym derivative.  
Moreover, it is easily seen that for 
$\propkerm{m}(\xvec{m}, \cdot)$-almost all $\xvec{m+1} \in \Xp{m+1}$, 
$$ 
\rnm{m}(\xvec{m},\xvec{m+1})=\prod_{k=n_m}^{n_{m+1}-1}\rn{k}(x_k,x_{k+1}).  
$$ 
Finally, we define the proposal probability measure $\Xinitpropm =\nu\tensprod\propker{0}\tensprod\cdots\tensprod\propker{n_0-1}$. This measure is absolutely continuous with respect to $\Xinitm$, and we let $\rnm{-1}$ be the Radon--Nikodym derivative of $\Xinitm$ with respect to $\Xinitpropm$. It is easily shown that  $\rnm{-1}(\xvec{0})=\rn{-1}(x_0)\prod_{k=0}^{n_0-1}\rn{k}(x_k,x_{k+1})$ for $\Xinitpropm$-almost all $\xvec{0} \in \Xp{0}$.

Algorithm \ref{algo:pfmulti} shows one iteration of the APF for the extended model, which is initialised by drawing $(\epartm{0}{i})_{i=1}^N$ from $\Xinitpropm^{\varotimes N}$ and letting $\wgtm{0}{i} \gets \rnm{-1}(\epartm{0}{i})$ for all $i\in\intvect{1}{N}$. \update{ Proposition~\ref{prop:original:vs:extended} connects the output of this algorithm to that of Algorithm~\ref{algo:adaptivePF} (cf. \cite[Proposition~C.1]{mastrototaro:olsson:alenlov:2021}, which states a similar result in the context of particle-based additive smoothing).}

\begin{algorithm}[htb]
	\caption{Bootstrap particle with systematic selection in the extended model.}\label{algo:pfmulti}
	\begin{algorithmic}[1] 
		\Statex \textbf{Data}: $(\epartm{m}{i},\wgtm{m}{i})_{i=1}^\N$.
		\Statex \textbf{Result}: $(\epartm{m+1}{i},\wgtm{m+1}{i},\I{m+1}{i})_{i=1}^\N$
		\For{$i=1\rightarrow\N$}
		\State draw $\I{m+1}{i} \sim \catdist((\wgtm{m}{\ell}\amm{m}(\epartm{m}{\ell}))_{\ell=1}^\N)$;
		\State draw $\epartm{m + 1}{i} \sim \propkerm{m}(\epartm{m}{\I{m + 1}{i}},\cdot)$;
		\State set $\displaystyle \wgtm{m+1}{i} \gets \rnm{m}(\epartm{m}{\I{m + 1}{i}},\epartm{m + 1}{i})/\amm{m}(\epartm{m}{\I{m + 1}{i}})$;
		\EndFor
	\end{algorithmic}
\end{algorithm}

\begin{proposition} \label{prop:original:vs:extended}
Let $(\res{n})_{n \in \nset}$ be a deterministic selection schedule and let $(n_m)_{m \in \nset}$ be the induced selection times. Furthermore, let $(\epart{n_m}{i}, \wgt{n_m}{i})_{i = 1}^\N$, $m \in \nset$, be a subsequence of weighted samples generated by Algorithm~\ref{algo:adaptivePF} for the original model and let $(\epartm{m}{i}, \wgtm{m}{i})_{i = 1}^\N$, $m \in \nset$, be weighted samples generated by Algorithm~\ref{algo:pfmulti} for the extended model. Then for every $m \in \nset$, 
$$
(\proj_{n_m}^m(\epartm{m}{i}), \wgtm{m}{i})_{i = 1}^\N \stackrel{\mathcal{D}}{=} (\epart{n_m}{i}, \wgt{n_m}{i})_{i = 1}^\N. 
$$
 \end{proposition}

\begin{proof}
We proceed by induction. Assume that we have generated a sample $(\epart{n_m}{i}, \wgt{n_m}{i})_{i = 1}^\N$ by means of $n_m$ iterations of Algorithm~\ref{algo:adaptivePF}, and that the claim holds true for this sample. We now examine the output of iteration $n_{m + 1}$. Since we know that $\rho_{n_m} = 1$, the sample at time $n_m + 1$ is produced by selection and mutation; thereafter, selection is not performed until time $n_{m+1}$ (since $\rho_k = 0$ for all $k \in \intvect{n_m + 1}{n_{m + 1} - 1}$). Hence, each particle path $\epart{n_m + 1:n_{m + 1}}{i}$ will be generated according to 
\begin{equation} \label{eq:path:distribution-1}
\epart{n_m + 1:n_{m + 1}}{i} \sim \propker{n_m} \varotimes \cdots \varotimes  \propker{n_m - 1}(\epart{n_m}{\I{n_{m}+1}{i}}, \cdot) 
\end{equation}
and assigned the importance weight 
\begin{equation} \label{eq:weights-1}
\wgt{n_{m + 1}}{i} = \frac{\rn{n_m}(\epart{n_m}{\I{n_{m}+1}{i}},\epart{n_m+1}{i})}{\am{n_m}(\epart{n_m}{\I{n_{m}+1}{i}})}\prod_{k=n_{m}+1}^{n_{m+1}-1}\rn{k}(\epart{k}{i},\epart{k+1}{i}), 
\end{equation}
where 
\begin{equation} \label{eq:I:dist-1}
\I{n_{m}+1}{i} \sim \catdist((\wgt{n_m}{\ell}\am{n_m}(\epart{n_m}{\ell}))_{\ell = 1}^\N).
\end{equation}
Now, on the other hand, by applying one iteration of Algorithm~\ref{algo:pfmulti} to the sample $(\epartm{m}{i}, \wgtm{m}{i})_{i = 1}^\N$ we obtain path-particles $\epartm{m + 1}{i} = \epart{n_m + 1:n_{m+ 1}}{i}$, $i \in \intvect{1}{\N}$, with distribution  
\begin{equation} \label{eq:eq:path:distribution-2}
\epartm{m + 1}{i} = \epart{n_m + 1:n_{m+ 1}}{i} \sim 
\propkerm{m}(\epartm{m}{\I{m+1}{i}}, \cdot) = \propker{n_m} \varotimes \cdots \varotimes \propker{n_m - 1}(\proj_{n_m}^m(\epartm{m}{\I{m+1}{i}}), \cdot),
\end{equation}
whose associated weights are
\begin{align} 
\wgtm{m+1}{i} &= \frac{\rnm{m}( \epartm{m}{\I{m+1}{i}},\epartm{m + 1}{i})}{\amm{m}(\epartm{m}{\I{m+1}{i}})}=  \frac{\rn{n_m}(\epart{n_m}{\I{m+1}{i}},\epart{n_m+1}{i})}{\am{n_m}(\epart{n_m}{\I{m+1}{i}})}\prod_{k=n_{m}+1}^{n_{m+1}-1}\rn{k}(\epart{k}{i},\epart{k+1}{i})\label{eq:weights-2}
\end{align}
and where 
\begin{equation} \label{eq:I:dist-2}
\I{m+1}{i} \sim \catdist((\wgtm{m}{\ell}\amm{m}(\epartm{m}{\ell}))_{\ell = 1}^\N) = 
 \catdist((\wgtm{m}{\ell}\am{n_m}(\proj_{n_m}^m(\epartm{m}{\ell}))_{\ell = 1}^\N).
\end{equation}
Finally, by comparing \eqref{eq:path:distribution-1} and \eqref{eq:eq:path:distribution-2}, \eqref{eq:weights-1} and \eqref{eq:weights-2}, \eqref{eq:I:dist-1} and \eqref{eq:I:dist-2}, and applying the induction hypothesis,  
$$
(\proj_{n_m+1}^{m + 1}(\epartm{m + 1}{i}), \wgtm{m + 1}{i})_{i = 1}^\N \stackrel{\mathcal{D}}{=} (\epart{n_{m + 1}}{i}, \wgt{n_{m + 1}}{i})_{i = 1}^\N. 
$$
The base case $m = 0$ is established similarly. This completes the proof. 
\end{proof}

Thus, in the case of a deterministic---but possibly irregular---resampling schedule, the APF may be reinterpreted as a particle model with systematic resampling operating on the auxiliary model described above. As the CLT in Proposition~\ref{prop:clt} is a general result, valid for arbitrary models and APFs (with systematic resampling), it holds also for the extended model and Algorithm~\ref{algo:pfmulti}, and the asymptotic normality of the output of Algorithm~\ref{algo:pfmulti} follows immediately. This finding is summarised by the following proposition. 

\begin{proposition} \label{prop:cltext}
	Assume that the functions $\rnm{-1}$, $(\rnm{m}/\amm{m})_{m\in\nset}$, and $(\amm{m})_{m\in\nset}$ are all bounded. Then for every $m\in\nset$ and $\funcm{m}\in\bmf{\Xpalg{m}}$, as $N\to\infty$,
	\begin{equation}\label{eq:CLTmulti}
	\sqrt{N}\left(\sum_{i=1}^N\frac{\wgtm{m}{i}}{\wgtsumm{m}}\funcm{m}(\epartm{m}{i})-\postm{m}\funcm{m}\right)\convd \boldsymbol{\sigma}_m(\funcm{m})Z,
	\end{equation}
	where $Z$ is standard normally distributed and $\boldsymbol{\sigma}_m^2(\funcm{m}) \eqdef \boldsymbol{\sigma}_{0,m}^2(\funcm{m})$, with, for $\ell\in\intvect{0}{m}$, 
	\begin{multline}\label{eq:truncVarExt}
	\asvarm{\ell,m}(\funcm{m})\eqdef \frac{\Xinitm(\rnm{-1}\{\lkm{0:m-1}(\funcm{m}-\postm{m} \funcm{m})\}^2)}{(\Xinitm\lkm{0:m-1}\1{\Xp{m}})^2}\1{\{\ell=0\}}
	\\+\sum_{k=(\ell-1)\vee 0}^{m-1}\postm{k}\amm{k}\frac{\postm{k}\lkm{k}(\amm{k}^{-1}\rnm{k}\{\lkm{k+1:m-1}(\funcm{m}-\postm{m} \funcm{m})\}^2)}{(\postm{k}\lkm{k:m-1}\1{\Xp{m}})^2}.
	\end{multline}
\end{proposition}

Proposition~\ref{prop:cltext} implies that we may obtain a CLT also in the case where Algorithm~\ref{algo:adaptivePF} is driven by a deterministic resampling schedule $(\res{n})_{n\in\nset}$. To conclude formally this argument, consider the output of Algorithm~\ref{algo:adaptivePF} after an arbitrarily chosen number $n$ of iterations; even though $n$ is generally not a resampling time, we may, without loss of generality, assume that it is so (since the distribution of the particle sample at a given time point does not depend on whether selection will be performed in a subsequent
iteration of the algorithm). In particular, in the extended model we may let $\Xp{\nres{n}} = \set{X}_{n_{\nres{n}-1}+1}\times\cdots\times\set{X}_n$; then by Proposition~\ref{prop:original:vs:extended}, 
$$
(\proj_{n}^{\nres{n}}(\epartm{\nres{n}}{i}), \wgtm{\nres{n}}{i})_{i = 1}^\N \stackrel{\mathcal{D}}{=} (\epart{n}{i}, \wgt{n}{i})_{i = 1}^\N. 
$$
For any function $\func{n} \in \bmf{\set{X}_{n}}$ we may define $\funcm{\nres{n}} : \Xp{\nres{n}} \ni \xvec{\nres{n}} \mapsto \func{n}(\proj_n^{\nres{n}}(\xvec{\nres{n}}))=\func{n}(x_{n})$. It is straightforward to check that for a so-defined extended function it holds that $\post{n}\func{n}=\postm{\nres{n}}\funcm{\nres{n}}$. Thus, under Assumption~\ref{assum:bounded}, Proposition~\ref{prop:cltext} implies that, as $N\to \infty$,
\begin{equation}\label{eq:CLTdetres}
\sqrt{N}\left(\post[N]{n}\func{n}-\post{n}\func{n}\right)\convd \sigma_n\langle \res{0:n-1}\rangle(\func{n})Z,
\end{equation}
where $Z$ is standard normally distributed and, since $\post[N]{n}\func{n}=\postm[N]{\nres{n}}\funcm{\nres{n}}$ and $\post{n}\func{n}=\postm{\nres{n}}\funcm{\nres{n}}$, the asymptotic variance $\asvar{n}\langle \res{0:n-1}\rangle(\func{n})$ equals $\asvarm{\nres{n}}\langle \res{0:n-1}\rangle(\funcm{\nres{n}})$; here we have added $\res{0:n-1}$ to the notation in order to highlight that the extended model under consideration is governed by the given selection schedule.

\subsection{Proof of Theorem~\ref{lemma:CLTadaptive}}
\label{sec:proof:lemma:CLTadaptive}
The following proof resembles closely the proof of \cite[Corollary~3.7]{mastrototaro:olsson:alenlov:2021}.
\begin{proof}
	Let $\set{R}_n \eqdef \{0,1\}^{n}$ be the set of binary sequences of length $n$. To all $\res{0:n-1} \in \set{R}_n$ we associate independent realisations $(\epart{n}{i}, \wgt{n}{i})_{i=1}^\N$ of Algorithm~\ref{algo:adaptivePF}, each realisation being driven by the deterministic selection schedule governed by the corresponding $\res{0:n-1}$, and let $\func{n}^\N \langle \res{0:n-1}\rangle \eqdef \wgtsum{n}^{-1}\sum_{i=1}^{n}\wgt{n}{i}\func{n}(\epart{n}{i})$ denote the corresponding filter estimate. Then for every $\N \in \nset^\ast$, 
	\begin{equation}
		\sqrt{\N}\left(\post[\N]{n} \func{n} 
		-\post{n}\func{n}\right) \\ 
		\eqdist\sum_{\res{0:n-1}\in\set{R}_n}\sqrt{\N}\left(\func{n}^\N \langle \res{0:n-1}\rangle-\post{n}\func{n}\right)\prod_{m=0}^{n-1}\1{\{\res[\N]{m}=\res{m}\}}. \label{eq:sumclt}
	\end{equation}
	By Lemma~\ref{lemma:ess},
	\begin{equation}\label{eq:convInd}
		\prod_{m=0}^{n-1}\1{\{\res[\N]{m}=\res{m}\}}\convp \1{\{\res[\alpha]{0:n-1}=\res{0:n-1}\}},
	\end{equation}
	and by Slutsky's Lemma and \eqref{eq:CLTdetres}, all terms in the sum \eqref{eq:sumclt} tends to zero except one, which converges weakly to $\sigma_n\langle \res[\alpha]{0:n-1}\rangle (\func{n}) Z$. This completes the proof. 
\end{proof}

\subsection{Proof of Corollary \ref{cor:convAdaptive}}\label{subsec:proofCor}
We first show the consistency of the variance estimates provided by Algorithm~\ref{algo:adaptive2} in the case of a deterministic resampling schedule. 

\begin{lemma}\label{lemma:convVarDetRes}
	Let Assumption~\ref{assum:1} hold. For every $n\in\nset$ and functionals $\func{m}\in\bmf{\alg{X}_m}$, $m\in\intvect{1}{n}$, let $(\lambda_m)_{m=1}^n$ be the lags produced by $n$ iterations of Algorithm~\ref{algo:adaptive2} driven by some deterministic selection schedule $(\res{n})_{n \in \nset}$. Then, as $N\to\infty$, it holds that $\lambda_n\convp \nres{n}$ and $\asvarp{n,\lambda_n}(\func{n})\convp \asvar{n}\langle\res{0:n-1}\rangle(\func{n})$. 
\end{lemma}
\begin{proof}

We proceed by induction, assuming that this holds true for some $n-1$ with $n\in\nsetpos$. Along the lines of the proof of Theorem~\ref{thm:conv}, we show that $\lambda_n$ produced by Algorithm~\ref{algo:adaptive2} converges in probability to $\nres{n}$. We consider separately the two cases $\res{n-1}=0$ and $\res{n-1}=1$. In the former, resampling is not performed at time $n-1$; thus, $\lambda_n=\lambda_{n-1}$ and $\nres{n}=\nres{n-1}$, which implies that 
$$
	\prob(\lambda_n=\nres{n})=\prob(\lambda_{n-1}=\nres{n-1}) \to 1
$$
as $N\to \infty$ by the induction hypothesis. In the latter case $\res{n-1}=1$, when resampling is triggered, $\nres{n}=\nres{n-1}+1$ while $\lambda_n$ is determined by Line~10 in Algorithm~\ref{algo:adaptive2}. Thus, we may write 
\begin{align}
\prob(\lambda_n=\nres{n}) &= \prob(\lambda_n=\nres{n}, \lambda_{n-1}={\nres{n-1}})+\prob(\lambda_n=\nres{n}, \lambda_{n-1}<{\nres{n-1}}) \label{eq:first:term:probability:split} \\
&= \prob(\lambda_n=\nres{n}, \lambda_{n-1}={\nres{n-1}}),\label{eq:second:term:probability:split}  
\end{align}
 where the second term of \eqref{eq:first:term:probability:split} is zero since, necessarily, $\lambda_n\le \lambda_{n-1} + 1$ by construction. To treat the first term \eqref{eq:second:term:probability:split}, write 
\begin{equation} \label{eq:second:term:probability:split:developed}
\prob(\lambda_n=\nres{n}, \lambda_{n-1}={\nres{n-1}}) = \prob \left(\{\lambda_{n-1}={\nres{n-1}}\} \bigcap_{\lambda = 0}^{\nres{n} - 1} \{\asvarp{n,\nres{n}}(\func{n})\ge \asvarp{n,\lambda}(\func{n})\}\right). 
\end{equation}
In order to show that \eqref{eq:second:term:probability:split:developed} tends to one, we consider the extended model with resampling times $(n_m)_{m\in\nset}$ induced by $(\res{n})_{n\in\nset}$. Without loss of generality we assume that $n$ is a resampling time. Then by Proposition~\ref{prop:original:vs:extended}, for every $\lambda\in\intvect{0}{\nres{n}}$, 
$$
\asvarp{n,\lambda}(\func{n}) \eqdist N\sum_{i=1}^N\left(\sum_{j:\enoch{\nres{n}(\lambda),\nres{n}}{j}=i}\frac{\wgtm{\nres{n}}{j}}{\wgtsumm{\nres{n}}}\{\funcm{\nres{n}}(\epartm{\nres{n}}{j})-\postm[N]{\nres{n}}\funcm{\nres{n}}\}\right)^2, 
$$
where we have defined, as previously, $\funcm{\nres{n}} : \Xp{\nres{n}} \ni \xvec{\nres{n}} \mapsto \func{n}(\proj_n^{\nres{n}}(\xvec{\nres{n}}))$ (where $\Xp{\nres{n}} \eqdef \set{X}_{n_{\nres{n}-1}+1}\times\cdots\times\set{X}_n$). Furthermore, by noting that  
$$
\frac{\rnm{m}(\xvec{m},\xvec{m+1}) \amm{m+1}(\xvec{m+1})}{\amm{m}(\xvec{m})}=\prod_{k=n_m}^{n_{m+1}-1}\frac{\rn{k}(x_k,x_{k+1})\am{k+1}(x_{k+1})}{\am{k}(x_{k})}     $$
and
$$
\rnm{-1}(\xvec{0})\amm{0}(\xvec{0})=\rn{-1}(x_0)\am{0}(x_0)\prod_{k=0}^{n_0-1}\frac{\rn{k}(x_k,x_{k+1})\am{k+1}(x_{k+1})}{\am{k}(x_{k})},
$$
 we conclude that once Assumption~\ref{assum:1} is satisfied for the original model, then it is also satisfied for the extended one. Thus, Theorem~\ref{thm:consistency-fixed-lag} implies that for every $\lambda\in\intvect{0}{\nres{n}}$, 
\begin{equation}\label{eq:convTruncMult}
\asvarp{n,\lambda}(\func{n})\convp \asvarm{\nres{n}(\lambda),\nres{n}}\langle \res{0:n-1}\rangle(\funcm{\nres{n}}),
\end{equation}
where we have included $\res{0:n-1}$ in the notation to highlight that the extended model under consideration is determined by the given selection schedule. By  \eqref{eq:truncVarExt} it holds that for $\lambda \in \intvect{0}{\nres{n} - 1}$, 
$$
\asvarm{\nres{n}\langle\lambda\rangle,\nres{n}}\langle \res{0:n-1}\rangle(\funcm{\nres{n}})\le\asvarm{\nres{n}\langle\nres{n}\rangle,\nres{n}}\langle \res{0:n-1}\rangle(\funcm{\nres{n}})=\asvarm{\nres{n}}\langle \res{0:n-1}\rangle(\funcm{\nres{n}});
$$ 
thus, using (\ref{eq:second:term:probability:split}--\ref{eq:convTruncMult}) and the induction hypothesis we may conclude that $\prob(\lambda_n=\nres{n}) \to 1$ as $N\to\infty$. 
    
Finally, we show that for every $\varepsilon>0$, $\prob(|\asvarp{n,\lambda_n}(\func{n})- \asvar{n}\langle\res{0:n-1}\rangle(\func{n})| \geq 2\varepsilon)$ tends to zero as $\N \to \infty$. Recalling that $\asvar{n}\langle \res{0:n-1}\rangle(\func{n}) = \asvarm{\nres{n}} \langle \res{0:n-1}\rangle(\funcm{\nres{n}})$, we obtain the bound 
\begin{multline}
\prob(|\asvarp{n,\lambda_n}(\func{n})- \asvar{n}\langle\res{0:n-1}\rangle(\func{n})|\geq2\varepsilon)\\ 
\le \prob(|\asvarp{n,\lambda_n}(\func{n})- \asvarp{n,\nres{n}}(\func{n})|>\varepsilon)+\prob(|\asvarp{n,\nres{n}}(\func{n})- \asvarm{\nres{n}} \langle \res{0:n-1}\rangle(\funcm{\nres{n}})| \geq \varepsilon),
\end{multline}
where the second term tends to zero as $N\to \infty$ by \eqref{eq:convTruncMult}. For the first term it holds that 
\begin{multline}
\prob(|\asvarp{n,\lambda_n}(\func{n})- \asvarp{n,\nres{n}}(\func{n})| \geq \varepsilon)\\
\leq \prob(\asvarp{n,\lambda_n}(\func{n})\ne \asvarp{n,\nres{n}}(\func{n})) \le \prob(\lambda_n \ne \nres{n})\to 0,
\end{multline}
as $N\to\infty$. 
    
The proof is completed by noting that the base case holds true, since $\lambda_0=0$ and $\asvarp{0,0}(\func{0})\convp \asvar{0}(\func{0})$ for all $\func{0}\in\bmf{\alg{X}_0}$. 

\end{proof}
We are now ready to prove Corollary~\ref{cor:convAdaptive}.
\begin{proof}[Proof of Corollary~\ref{cor:convAdaptive}]
    Following the lines of the proof of Lemma~\ref{lemma:CLTadaptive}, we let again $\set{R}_n = \{0,1\}^n$ be the set of binary sequences of length $n$. For all $\res{0:n-1} \in \set{R}_n$, let $\asvarp{n,\lambda_n} \langle \res{0:n-1}\rangle(\func{n})$ be independent variance estimators obtained on the basis of independent realisations of Algorithm~\ref{algo:adaptive2}, each realisation being driven by the corresponding selection schedule $\res{0:n-1}$. Then for every $\N \in \nset^\ast$, 
	\begin{equation}\label{eq:sumVarEst}
		\asvarp{n,\lambda_n}\langle \res[N]{0:n-1}\rangle(\func{n})
		\eqdist\sum_{\res{0:n-1}\in\set{R}_n}\asvarp{n,\lambda_n}\langle \res{0:n-1}\rangle(\func{n})\prod_{m=0}^{n-1}\1{\{\res[\N]{m}=\res{m}\}}.
	\end{equation}
	Now, by \eqref{eq:convInd} and Lemma~\ref{lemma:convVarDetRes}, all terms of \eqref{eq:sumVarEst} tend to zero as $\N \to \infty$ except one which converges in probability to $\asvar{n}\langle \res[\alpha]{0:n-1}\rangle(\func{n})$. This completes the proof. 
\end{proof}